\documentclass[a4paper,11pt]{article}
\pdfoutput=1

\usepackage{jheppub}
\usepackage[T1]{fontenc} 
\usepackage{amsmath,amssymb,amsfonts,mathrsfs}
\usepackage{dsfont}
\usepackage{tikz}
\usepackage{tikz-cd}
\usepackage{todonotes}
\usetikzlibrary{cd}
\usepackage[all]{xy}
\usetikzlibrary{shapes,arrows}
\usepackage{mathtools}
\usepackage{amsthm}

\usepackage{booktabs}
\usepackage{enumitem}

\def\bZ{\mathbb{Z}}

\def\bC{\mathbb{C}}

\def\Tr{\textrm{Tr}}
\def\A{\mathbf{A}}

\def\bp{\begin{pmatrix}}
\def\ep{\end{pmatrix}}

 \newcommand{\slTwo}{\mathfrak{sl}_2}

\let\Tr\relax
\DeclareMathOperator{\Tr}{Tr}

\newtheorem{theorem}{Theorem}[section]
\newtheorem{lemma}[theorem]{Lemma}
\newtheorem{proposition}[theorem]{Proposition}
\newtheorem{remark}[theorem]{Remark}

\usepackage{graphicx}
\usepackage[utf8]{inputenc}
\usepackage{url}
\usepackage{hyperref}
\usepackage{braket}
\usepackage{xcolor}

\definecolor{red}{rgb}{1,0,0}
\long\def\comment#1{}
\usepackage{chessfss}
\pdfpageattr {/Group << /S /Transparency /I true /CS /DeviceRGB>>}
\hypersetup{colorlinks=true,linkcolor=blue,citecolor=blue}
\usepackage{bm}
\usepackage{float}
\usepackage{subfig}
\usepackage{scalefnt}
\usepackage{wrapfig}
\usepackage{fancybox}
\usepackage{multirow}
\usepackage{adjustbox}
\usepackage[toc,page]{appendix}

\title{\boldmath Quantum Flat Connections, KZ equations, and Integrability}
\author[a,b]{Sibasish Banerjee,}
\author[c,d]{Babak Haghighat,}
\author[e]{Anouchah Latifi}

\affiliation[a]{Institut des Hautes \'Etudes Scientifiques (IHES), Bures-sur-Yvette 91440, France}
\affiliation[b]{Vivatsgasse 7, Max Planck Institute for Mathematics, Bonn - 53111, Germany}
\affiliation[c]{Yau Mathematical Sciences Center, Tsinghua University, Beijing, 100084, China}
\affiliation[d]{Beijing Institute of Mathematical Sciences and Applications (BIMSA), Huairou District, Beijing 101408, China}
\affiliation[e]{Qom University of Technology (QUT), Qom, Iran}

\abstract{
$\mathcal{N}=2$ supersymmetric Yang-Mills theories are described in terms of a Hitchin system over a Riemann surface $\mathcal{C}$. Focusing on strongly coupled Argyres-Douglas theories, we show that the corresponding flat bundle over $\mathcal{C}$ can be quantized such that the resulting quantum flat connection is integrable. For $\mathfrak{sl}_2$, the quantum connection takes values in $\mathfrak{gl}_2(\mathcal{A})$ where $\mathcal{A}$ is an associative algebra which we explicitly describe for the cases of Painlev\'e I, II and IV. Moreover, we find that the quantum connection is equivalent to irregular versions of Knizhnik-Zamolodchikov (KZ) connections. Utilizing a suitable gauge transformation, one can show that the corresponding KZ equations give rise to BPZ equations.
}

\begin{document}

\maketitle

\section{Introduction}

Quantum gauge theories are intrinsically described by integrable systems when admitting $\mathcal{N}=2$ supersymmetry \cite{Gorsky:1995zq,Donagi:1995cf}. Among the main features of such systems is the existence of a flat connection taking values in the Lie algebra $\mathfrak{g}$, known as the Higgs field \cite{Gaiotto:2009hg}, on a curve $\mathcal{C}$ describing the effective low energy vacua of the theory \cite{Seiberg:1994rs,Seiberg:1994aj}. Within this framework, the monodromy of the horizontal sections of the corresponding Higgs bundle around punctures of $\mathcal{C}$ parametrizes mass and more general flavor symmetries of the gauge theory \cite{Gaiotto:2009we}. The connection to integrable systems arises from isomonodromic deformations and in cases where the connection admits irregular singular points, the corresponding isomonodromic deformations give rise to the Painlev\'e hierarchy. One of the central observations of modern developments in supersymmetric quantum gauge theories was the identification of the Painlev\'e \textit{tau function} with a certain Fourier transform of the BPS partition function of the gauge theory \cite{GamayunIorgovLisovyy2012}. 

Another important development which unfolded in parallel to the above is the so-called BPS/CFT \cite{Nekrasov:2009rc,Nekrasov:2015wsu} or AGT \cite{Alday:2009aq} correspondence. Here one identifies the BPS partition function of the gauge theory with conformal blocks of Liouville (in the case where $\mathfrak{g} = \mathfrak{sl}_2$) or more generally Toda CFT (for general $\mathfrak{g}$) on the curve $\mathcal{C}$ with operator insertions at the punctures. The conformal blocks in question satisfy hypergeometric differential equations known as \textit{BPZ equations} \cite{Belavin:1984vu} when the gauge theory contains a surface defect \cite{Alday:2009fs}. More recently, the relationship of BPS partition functions to Knizhnik-Zamolodchikov (KZ) equations has been studied within the BPS/CFT correspondence \cite{Nekrasov:2017rqy,Jeong:2021vzo}.

One way to think about the connection between the above two perspectives is to view the BPZ equation and the Liouville conformal blocks as arising from a quantization of the corresponding integrable system. Within this viewpoint, the BPZ equation is a direct consequence of the quantization of the spectral curve or Seiberg-Witten curve of the gauge theory \cite{Nekrasov:2010ka,Maruyoshi:2010iu,Gukov:2011qp,Dumitrescu:2014}. Yet another perspective, which we want to advocate in the present paper, is to quantize the phase space of the flat connections themselves. This viewpoint has been taken in the past (see \cite{Nekrasov:2011bc,Teschner:2010je}, and more recently \cite{Gaiotto:2024bna}), but a  direct quantization of the corresponding integrable system has been less studied (though see \cite{Chekhov:2015qza} for the quantization of Painlev\'e monodromy manifolds). One of the main novelties of our approach is that we explicitly show that upon quantization the corresponding flat connection is promoted to a \textit{quantum flat connection} satisfying a quantum flatness condition. It would be interesting to connect our approach to the recent work \cite{Bonelli:2025owb} on quantum tau functions.

We focus on the cases of Painlev\'e I, II and IV which correspond to intrinsically strongly coupled Argyres-Douglas theories \cite{Argyres:1995jj} $H_0$, $H_1$ and $H_2$ \cite{Eguchi:1996vu}. These are superconformal field theories (SCFTs) admitting enhanced symmetry properties. We show that a direct quantization of the class-preserving deformation space of the Painlev\'e Lax pair gives rise to irregular versions of $\mathfrak{sl}_2$ KZ equations \cite{KNIZHNIK198483}. The study of irregular KZ equations was initiated in \cite{Resh-KZ,FMTV00,Jimbo_2008,Nagoya_2010}. More recently, these equations have been studied from the vantage point of their relations to BPZ equations \cite{Haghighat:2023vzu,Gu:2023plq}, irregular Kac-Moody representations \cite{Gukov:2024yxa,Haghighat:2025qxu}, and knot invariants \cite{Gu:2026mlf}. Solutions of irregular KZ equations correspond to $\mathfrak{sl}_2(\mathbb{C})$ WZW conformal blocks, see also \cite{Gaiotto:2013rk} on their appearance within the $H_3^+$ WZW model. At the same time (see above), we also know that the corresponding wavefunctions are also solutions of BPZ equations. This relation had already been observed in \cite{Ribault:2005wp} for regular KZ equations and our work is the first explicit generalization to the Argyres-Douglas case. Indeed, performing a quantum gauge transformation, we show how to transform the KZ equation to a BPZ equation arising from the energy-momentum tensor of irregular Virasoro conformal blocks \cite{Gaiotto:2012sf,Kanno:2013vi}. This nicely reproduces the BPZ equations studied in earlier work \cite{Maruyoshi:2010iu}.

\subsection{Summary of results}

The main results are as follows. In Section \ref{sec:2}, we revisit the connection between gauge theory and Hitchin systems and describe our quantization procedure. We then proceed in Section \ref{sec:KZ} to formulate irregular KZ equations in full generality where we focus on the $\mathfrak{sl}_2$ case. To this end, we define the loop algebra and affine currents which then allow us to formulate the full irregular KZ connection. We then show how our quantization procedure for the Painlev\'e I, II and IV connections give rise to irregular KZ equations with a suitable Whittaker module at infinity. This correspondence can be formulated as
\begin{equation}
    (\partial_z - A(z)) \Psi =0 \longrightarrow \nabla_i \widehat \Psi = 0.
\end{equation}
In the above, $A(z)$ is the $\mathfrak{sl}_2$-valued connection arising from the first Lax matrix and the KZ connection $\nabla_i$ is
\begin{equation}
    \nabla_i = \kappa \partial_{z_i} - \Omega_i(\{z\}),
\end{equation}
where $\Omega_i$ is valued in the enveloping algebra $U(\mathfrak{sl}_2).$

Let us now come to Section \ref{sec:integrability} which clarifies the connection to integrability and establishes the relation to BPZ equations. As a first result (see Lemma \ref{lem:forced_identity_component}), we show that the connection $A(z)$ upon quantization takes values in the algebra $\mathfrak{gl}_2(\mathcal{A})$ where $\mathcal{A}$ is an associative algebra. In fact, in the simplest setting of Painlev\'e I, 
\begin{equation}
    \mathcal{A} = W_{\kappa}(p,q), 
\end{equation}
where $W_{\kappa}(p,q)$ denotes the Weyl algebra \eqref{eq:Weyl} arising from the quantization of the canonical coordinate and its conjugate momentum. To study the Painlev\'e II and IV cases, we first calculate the tangent space of class-preserving deformations. The main results for Painlev\'e II are contained in Propositions \ref{prop:PII-tangent} and \ref{prop:PII-completeness} while Painlev\'e IV is treated in Appendix \ref{sec:PIV-class-preserving}. We find that there exists a suitable choice of coordinates, such that the tangent space is parametrized by infinitesimal deformations $(\delta p, \delta q, \delta \theta, \delta u)$ together with a symplectic form
\begin{equation}
    \omega = dp \wedge dq - d\theta \wedge d\log u.
\end{equation}
Quantum flat connections take values in $\mathfrak{gl}_2(\mathcal{A})$ where $\mathcal{A}$ is now given by
\begin{equation}
    \mathcal{A} = W_\kappa(p,q) \otimes W_{-\kappa}(\theta,\log u).
\end{equation}
In Section \ref{sec:painleve_op_flatness} we show that the quantized flat connections satisfy a quantum flatness condition, i.e. that classical flatness also holds at the quantum level. Finally, in Section \ref{sec:Tgauge}, it is shown that there exists a gauge transformation $S$ with
\begin{equation}
    A^S = S A S^{-1} +(\partial_z S) S^{-1},
\end{equation}
such that the gauge transformed connection takes the form
\begin{equation}
    A^S = \left(\begin{array}{cc}
        0 & 1 \\
        T(z) & 0
    \end{array}\right)~.
\end{equation}
In the quantized setting, this gives rise to the BPZ equation
\begin{equation}
    \left(\partial_z^2 + \widehat{T}(z)\right) \psi = 0.
\end{equation}

\section{Quantum flat connections from Gauge theory}
\label{sec:2}

This section reviews the origins of quantum flat connections within the framework of $\mathcal{N}=2$ supersymmetric gauge theories and their compactification to Hitchin systems. By tracing the dimensional reduction from six-dimensional theories down to three dimensions, it details how the moduli space of Hitchin equations connects to the Coulomb branch of four-dimensional theories, culminating in the construction of classical integrable systems and the Seiberg-Witten curve.

\subsection{Gauge theory and Hitchin systems: a review}

In this section we aim to present a review of basic notions in $\mathcal{N}=2$ supersymmetric gauge theories and their relation to Hitchin systems. Let us start with the six-dimensional $\mathcal{N}=(2,0)$ theory of type $A_1$ which is then compactified on a Riemann surface $\mathcal{C}_{g,n}$. This gives rise to a family of four-dimensional $\mathcal{N}=2$ gauge theories known as class $\mathcal{S}$. Upon further compactifying on a circle $S^1$, one obtains a three-dimensional theory with $\mathcal{N}=4$ supersymmetry which is a sigma model with target space the moduli space $\mathcal{M}$. $\mathcal{M}$ is topologically a Jacobian fibration of the Seiberg-Witten curve of the uncompactified 4d theory over its moduli space of quantum vacua (i.e. the $u$-place in the case of pure $SU(2)$ gauge theories). 

Yet from another perspective, $\mathcal{M}$ admits a description as the moduli space of solutions to a Hitchin system. To see this, note that instead of first compactifying the 6d $(2,0)$
theory on $S^1$, one can alternatively first compactify on the circle and then on $\mathcal{C}_{g,n}$. The first step produces a five dimensional $\mathcal{N}=1$ $SU(2)$ gauge theory which on the Coulomb branch admits vector multiplets labeled by Cartan generators of the gauge group. In the second step the components $(A_z,A_{\bar z})$ of the five-dimensional $SU(2)$ gauge field on $\mathcal{C}$ form together with a complex adjoint scalar field $\varphi_z$ (built out of two of the five adjoint scalars of the five-dimensional theory) a Higgs bundle on $\mathcal{C}$. $\mathcal{M}$ is the moduli space of solutions $(A,\varphi)$ to the Hitchin equations

\begin{align}
    F_{z\bar z} + [\varphi_z,\overline{\varphi}_{\bar z}] &= 0,  \nonumber \\
    \partial_{\bar z} \varphi_z + [A_{\bar z}, \varphi_z] &= 0, \label{eq:Hitchin} \\
    \partial_z \overline{\varphi}_{\bar z} + [A_z, \overline{\varphi}_{\bar z}] &= 0~.  \nonumber 
\end{align}
To connect to the Coulomb branch of four-dimensional $\mathcal{N}=2$ theories, one associates to every Higgs bundle a classical integrable system with spectral curve given by the Seiberg-Witten curve
\begin{equation}
    \Sigma_{\textrm{SW}}: \quad \textrm{Det}(y - \varphi_z) = 0 \quad \Rightarrow \quad y^2 = \frac{1}{2} \textrm{Tr}(\varphi_z^2)~.
\end{equation}
$\mathcal{M}$ is a Hyperk\"ahler manifold and for a given choice of complex structure $\zeta \in \mathbb{C}^*$, one can define a Hitchin connection $ \nabla = \partial - \mathcal{A}$, with
\begin{equation}
    \mathcal{A} = \frac{1}{\zeta} \varphi_z + A + \zeta \overline{\varphi}_{\bar z},
\end{equation}
such that the Hitchin equations \eqref{eq:Hitchin} are equivalent to the flatness of this connection. Now in the limit $\zeta \rightarrow 0$, $\mathcal{A}$ reduces to $\varphi_z$ and the flat connection becomes 
\begin{equation}
     \nabla \rightarrow \partial + \frac{\varphi_z}{\zeta}.
\end{equation}
One can reinterpret this as a flat connection on $\mathcal{C}_{g,n}$ given by
\begin{equation}
    \kappa \partial_z - \mathbf{A},
\end{equation}
where we are taking $\mathbf{A} = \varphi_z$ and $\kappa \rightarrow 0$. As we will see, $\kappa$ is a quantum parameter analogous to $\hbar$ and its nonzero value corresponds to quantization of the system in a certain sense we will explore.

\paragraph{Relation to Painlev\'e equations.}
There is a beautiful relation of the above discussion to the theory of Painlev\'e equations which we want to briefly review following \cite{Bonelli:2016qwg}. The above flat connection $\mathcal{A}$ can be used to build an integrable system of linear ODEs given by
\begin{equation} \label{eq:classA}
    \frac{d}{dz} \Psi(z) = \mathbf{A}(z) \Psi(z),
\end{equation}
where we now restrict to the $n$-punctured sphere $\mathcal{C}_{0,n}$, $\mathbf{A}(z) \in \mathfrak{sl}_2(\mathbb{C})$ complex and traceless. A specific choice of $\mathbf{A}(z)$ determines the global monodromy of the solution $\Psi(z)$. However, for a given global monodromy, there exists a whole family of $A(z;\{\vec{t}\})$ parametrized by moduli $\{\vec{t}\}$. This leads to considering isomonodromic, i.e. monodromy preserving, deformations $A(z,t)$ and $\Psi(z,t)$, giving rise to an extended system of differential equations
\begin{align}\label{eq:A and B}
    \partial_z \Psi(z,t) &= \mathbf{A}(z,t) \Psi(z,t), \\
    \partial_t \Psi(z,t) &= \mathbf{B}(z,t) \Psi(z,t),
\end{align}
together with the compatibility condition
\begin{equation} \label{eq:compatibility}
    \mathbf{A}_t(z,t) = \mathbf{B}_z(z,t) + [\mathbf{B}(z,t), \mathbf{A}(z,t)]. 
\end{equation}
We present a more detailed analysis of these equations in Section \ref{sec:integrability}. For now it suffices to note that the corresponding matrices $\mathbf{A}$ and $\mathbf{B}$ are classified by the Painlev\'e hierarchy. Then Table \ref{tab:PainleveGauge_Nf} and \ref{tab:PainleveGauge_H} show the relation between the connection $\mathbf{A}(z)$ and the Higgs field $\varphi_z$ of $\mathcal{N}=2$ $SU(2)$ gauge theories with $N_f$ flavors as well as Argyres-Douglas theories $H_0$, $H_1$, $H_2$.
\begin{table}[H]
    \centering
    \begin{tabular}{|c|c|c|c|c|c|}
        \hline 
        VI & V & $\textrm{III}_1$ & $V_{deg}$ & $\textrm{III}_2$ & $\textrm{III}_3$ \\ \hline
        $N_f = 4$ & $N_f = 3$ & $N_f =2$ (1st) & $N_f =2$ (2nd) & $N_f = 1$ & $N_f = 0$ \\ \hline
    \end{tabular}
    \caption{Cases with $N_f = 0,\dots,4$.}
    \label{tab:PainleveGauge_Nf}
\end{table}

\begin{table}[H]
    \centering
    \begin{tabular}{|c|c|c|c|}
        \hline 
        IV & $\textrm{II}_{\textrm{JM}}$ & $\textrm{II}_{\textrm{FN}}$ & I \\ \hline
        $H_2$ & $H_1$ (1st) & $H_1$ (2nd) & $H_0$ \\ \hline
    \end{tabular}
    \caption{Cases with $H_0$ to $H_2$.}
    \label{tab:PainleveGauge_H}
\end{table}

\paragraph{Example: Painlev\'e $\textrm{II}_{\textrm{JM}}$.}
The connection $\mathbf{A}(z)$ in this case is given by
\begin{equation} \label{eq:APIIJM}
    \mathbf{A} = A_1 + A_2 z + A_3 z^2 = \left(\begin{array}{cc}
        z^2 + p + t/2 & u(z - q)  \\
        -\frac{2}{u} (pz + \theta + p q) & - (z^2 + p + t/2) 
    \end{array}\right)~.
\end{equation}
This case corresponds to a single irregular puncture of degree $3$ at infinity. The compatibility condition \eqref{eq:compatibility} becomes in this case 
\begin{align}
    \dot{u} &= - q u,  \\
    \dot{q} &= p + q^2 + t/2, \\
    \dot{p} &= -2pq - \theta~.
\end{align}
The equation for the spectral curve/Seiberg-Witten curve is given by
\begin{equation} \label{eq:PIIJMSWc}
    y^2 = \frac{1}{2} \textrm{Tr}\mathbf{A}^2 = z^4 + t z^2 - 2\theta z + 2 \sigma_{II} + \frac{t^2}{4},
\end{equation}
where $\sigma_{II}(t)$ is the Hamiltonian function for the Painlev\'e $\textrm{II}$ case:
\begin{equation}
    \dot{p} = - \frac{\partial}{\partial q}\sigma_{II}, \quad \sigma_{II} = \frac{p^2}{2} + p q^2 + \frac{pt}{2} + q \theta.
\end{equation}

\subsection{The quantum flat connection}
\label{sec:qflcon}

Here, we want to turn to quantizing the flat connections $\mathbf{A}$ encountered in the previous section. To this end, note the following interesting observation. Let us consider the Painlev\'e II case \eqref{eq:APIIJM}. Setting 
\begin{equation} \label{eq:PIIphasespace}
    (y,z) = (p + q^2 + \frac{t}{2}, q),
\end{equation}
one easily sees that equation \eqref{eq:PIIJMSWc} is satisfied. Thus for each value of $t$, this point is a distinguished point on the Seiberg-Witten curve. This point can be thought of as the position of a probe brane/operator. Moreover, since $y$ is a coordinate on the fibers of the cotangent bundle $T^* \mathcal{C}_{0,n}$, we also see that $(p,q)$ serve as coordinates on the phase space of the probe brane. From Hamiltonian mechanics, we then see that 
\begin{equation}
    \omega \equiv dp \wedge dq,
\end{equation}
endows this phase space with a symplectic pairing. To study the quantum mechanics of this probe brane, we quantize the system by imposing the commutation relations \cite{Aganagic:2003qj}
\begin{equation} \label{eq:quant}
    [\hat{p},\hat{q}] = \kappa.
\end{equation}
This is equivalent to introducing a surface defect in the four-dimensional gauge theory \cite{Alday:2009fs,Drukker:2010jp,Frenkel:2015rda}. Upon quantization, the equation for the Seiberg-Witten curve \eqref{eq:PIIJMSWc} translates to the BPZ equation \cite{Belavin:1984vu} for the Liouville conformal block associated to $\mathcal{C}_{0,n}$ with one degenerate operator at position $z$. In the case of Painlev\'e II considered above, $n=1$ as there is only one irregular operator of degree $3$ at infinity. The corresponding BPZ (also known as \textit{quantum curve}) equation reads then
\begin{equation}
\label{eq:BPZfull}
    \left(\kappa^2 \partial_z^2 - \frac{1}{2} \Tr \A^2(z;\hat{p},\hat{q})+ \mathcal{O(\kappa)}\right) \langle \Phi_{1,2}(z) \Phi_{1,2}(q)I^{(3)}(\infty)\rangle = 0.
\end{equation}
This is the irregular version of the BPZ/quantization correspondence developed in \cite{Teschner:2010je}. Namely, an irregular operator $I^{(3)}$ is placed at infinity. Equation \eqref{eq:BPZfull} can be viewed as the quantization of the classical curve $y^2 = \frac{1}{2} \Tr \mathbf{A}(z;p,q)^2$ when promoting $y$, $z$, $p$ and $q$ to operators,
\begin{equation} \label{eq:quantfull}
    [\hat y, \hat z] = \epsilon_1, \quad [\hat p, \hat q] = \epsilon_2,
\end{equation}
and choosing the slice $\epsilon_1 = 1$, $\epsilon_2 = \kappa$, on which we shall focus in the remainder of this paper. The more general quantization condition \eqref{eq:quantfull} will be treated elsewhere.

The remaining question is about the interpretation of the quantized version of equation \eqref{eq:classA}, namely
\begin{equation}
    \kappa \partial_z \widehat \Psi = \A(z;\hat{p},\hat{q}) \widehat \Psi,
\end{equation}
where we now treat $\A$ as a quantum operator. In the next section we will show that this equation is an irregular version of the celebrated Knizhnik-Zamolochikov (KZ) equations.

\section{KZ equations}
\label{sec:KZ}

This section provides a comprehensive formulation of both regular and irregular Knizhnik-Zamolodchikov (KZ) equations, focusing specifically on the $\mathfrak{sl}_2$ case. It builds the necessary mathematical machinery, defining loop algebras, affine currents, and higher evaluation representations acting on irregular Kac-Moody modules, to explicitly construct the irregular KZ connection and its associated Hamiltonians.

\subsection{Irregular KZ equations: a review}

Irregular KZ equations were first studied in \cite{Resh-KZ}. Here we present a short review of the subject starting from regular KZ equations \cite{KNIZHNIK198483}, where we focus on the $\mathfrak{sl}_2$ case. To begin with, let us fix our notation. Consider a Lie algebra $\mathfrak{g}$ together with
\begin{equation}
    L\mathfrak{g} = \bigoplus_{n \in \bZ} \mathfrak{g}[n], \quad \mathfrak{g} \simeq \mathfrak{g} t^n,
\end{equation}
and let $I_a$ be an orthonormal basis in $\mathfrak{g}$. The \textit{Loop algebra} $L\mathfrak{g}$ together with the central extension $\bC c$ defines the \textit{affine} Lie algebra $\widehat{\mathfrak{g}}$. We have the following evaluation representation
\begin{equation}
    \pi^{(0)}_z : \mathfrak{g}[n] \rightarrow \mathfrak{g} z^n, \quad x[n] \mapsto x z^n,
\end{equation}
where $\mathfrak{g}$ is understood as $\textrm{End}(V)$ for a $\mathfrak{g}$-module $V$ and $x$ is an element of $\mathfrak{g}$. Next, we define currents
\begin{equation}
    I^+_a(u) = \sum_{n \geq 0} I_a[n] u^{-n-1}, \quad
    I^-_a(u) = \sum_{n < 0} I_a[n] u^{-n-1} = \sum_{n \geq 0} I_a[-n-1] u^n.
\end{equation}
Applying the evaluation representation to these currents, one obtains
\begin{align}
    \pi^0_z : I^+_a(u) &\mapsto \sum_{n \geq 0} I_a z^n u^{-n-1} = \frac{I_a}{u-z},~|z| \ll |u|, \\
    \pi^0_z : I^-_a(u) &\mapsto \sum_{n \geq 0} I_a z^{-n-1} u^n = \frac{I_a}{z - u},~|u| \ll |z|.
\end{align}
We can also introduce \textit{higher} evaluation representations acting on degree $l$ irregular Kac-Moody modules (see \cite{Gukov:2024yxa}) via
\begin{equation}
    \pi^l_z : \quad I^+(u) \mapsto \frac{I^{(l)}_a}{(u-z)^{l+1}}, \quad I^-_a(u) \mapsto \frac{I^{(l)}_a}{(z-u)^{l+1}}~,
\end{equation}
where we have 
\begin{equation} \label{eq:gradedalg}
    [I^{(l)}_a, I^{(k)}_b] = \sum_c C_{abc} I^{(k+l)}_c,
\end{equation}
with $C_{abc}$ being the structure constants of the Lie algebra $\mathfrak{g}$. Next, following \cite{Resh-KZ,Jimbo_2008}, let us define the $r$-matrix
\begin{equation}
    r \equiv \sum_{n \geq 0} I_a[-n-1] \otimes I_a[n],
\end{equation}
and its projections
\begin{equation}
    r^{l,l'}(u,v) \equiv (\pi^l_u \otimes \pi^{l'}_v) r.
\end{equation}
One can show \cite{Resh-KZ} that
\begin{equation}
    r^{l,l'}(u,v) = \sum_a \sum_{s=0}^l \sum_{t=0}^{l'} (-1)^s \binom{t+s}{s} \frac{I^{(s)}_a \otimes I^{(t)}_a}{(u-v)^{s+t+1}},
\end{equation}
and that for 
\begin{equation} \label{eq:Omegai}
    \Omega_i(\{z\}) \equiv \sum_{j  eq i} r^{l_i l_j}_{ij}(z_i - z_j),
\end{equation}
one has 
\begin{equation}
    [\Omega_i(\{z\}), \Omega_j(\{z\})] = 0.
\end{equation}
In the above, the sub-indices $i$ and $j$ denote the action on the $i$th or $j$th entry in a tensor product of $\widehat{\mathfrak{g}}$-modules. The above result can be used to show that the connection
\begin{equation} \label{eq:nablaKZ}
     \nabla_i \equiv \kappa \frac{\partial}{\partial z_i} - \Omega_i(\{z\})
\end{equation}
is flat, i.e.
\begin{equation}
    [\nabla_i, \nabla_j] = 0. 
\end{equation}
We will henceforth restrict to $\mathfrak{g}=\mathfrak{sl}_2(\mathbb{C})$ with the standard matrix generators
$(H,X,Y)\in\slTwo$ acting on $\mathbb{C}^2$,
\begin{equation}
H=\begin{pmatrix}1&0\\0&-1\end{pmatrix},
\qquad
X=\begin{pmatrix}0&1\\0&0\end{pmatrix},
\qquad
Y=\begin{pmatrix}0&0\\1&0\end{pmatrix},
\label{eq:sl2-matrix-generators}
\end{equation}
and commutation relations
\begin{equation}
[H,X]=2X,
\qquad
[H,Y]=-2Y,
\qquad
[X,Y]=H.
\label{eq:sl2-matrix-relations}
\end{equation}
Then \eqref{eq:Omegai} becomes
\begin{equation} \label{eq:Omegai2}
    \Omega_i(\{z\}) = \sum_{s=0}^l \sum_{t=0}^{l'} (-1)^s \binom{t+s}{s} \frac{\Omega_{ij}^{(s,t)}}{(z_i - z_j)^{s+t+1}},
\end{equation}
where 
\begin{equation} \label{eq:Omegaij}
    \Omega_{ij}^{(s,t)} \equiv \frac{1}{2} H_i^{(s)} \otimes H_j^{(t)} + X_i^{(s)} \otimes Y_j^{(t)} + Y_i^{(s)} \otimes X_j^{(t)}.
\end{equation}

\subsection{Quantum flat connections as KZ equations}
\label{sec:PainleveKZmap}

We now want to show that the quantum flat connection defined in section \ref{sec:qflcon} is a deformation of the KZ connection. A similar observation, but with a totally different motivation, was made in \cite{Jimbo_2008}.

\paragraph{Painlev\'e $\textrm{I}$.} The Painlev\'e I connection is given by
\begin{equation}
    \mathbf{A}_I = -p H + (q^2 + zq + z^2 + t/2) X + (4z - 4q) Y.
\end{equation}
Using the fact that this connection describes one irregular operator at infinity and one degenerate operator at $z$ and comparing with \eqref{eq:Omegai2}, we can read off (setting $t=0$):
\begin{align}
    H_{\infty}^{(1)} = - 2 \hat p, &\quad X_{\infty}^{(1)} = -4 \hat q, \quad Y_{\infty}^{(1)} = \hat q^2, \\
    Y_{\infty}^{(2)} =  \hat q, &\quad X_{\infty}^{(2)} = 4, \quad H_{\infty}^{(2)} = 0 \\
    Y_{\infty}^{(3)} = 1, &\quad X_{\infty}^{(3)} = 0, \quad H_{\infty}^{(3)} = 0.
\end{align}
From this we can check that the following commutation relations hold:
\begin{align}
    [H_{\infty}^{(1)}, X_{\infty}^{(1)}] = 2 X_{\infty}^{(2)}, &\quad [X_{\infty}^{(1)}, Y_{\infty}^{(1)}] = H_{\infty}^{(2)} \\
    [H_{\infty}^{(1)}, Y_{\infty}^{(2)}] = - 2 Y_{\infty}^{(3)}, &\quad [H_{\infty}^{(1)}, X_{\infty}^{(2)}] = 2 X_{\infty}^{(3)}\\
    [H_{\infty}^{(1)}, Y_{\infty}^{(3)}] = -2 Y_{\infty}^{(4)}.
\end{align}
However, the following commutator fails:
\begin{equation}\label{faillure I}
    [H_{\infty}^{(1)}, Y_{\infty}^{(1)}] = - 4 q  \neq - 2 Y_{\infty}^{(2)}.
\end{equation}
One notices that this can be easily cured by adding the term $(zq + z^2)X$ to $\mathbf{A}$ which corresponds to setting
\begin{equation}\label{cure I}
    t/2 = zq + z^2.
\end{equation}
(As we will see shortly, this corresponds to choosing a specific slice that breaks mutual flatness but ensures KZ closure upon quantization.)
Note that this choice cannot be viewed as a deformation of the Lax pair as explicitly proven in Appendix \ref{subsec:KZ_shift_PI}. It rather corresponds to choosing a particular slice or hypersurface within the $(t,z)$ coordinate space. This slice breaks mutual flatness of $A$ and $B$, but ensures KZ closure upon quantization. That is, all commutators close and the loop algebra automatically truncates as follows:
\begin{equation}
    X_{\infty}^{(n)} = Y_{\infty}^{(n)} = H_{\infty}^{(n)} = 0 \quad \textrm{for} \quad n \geq 4.
\end{equation}

{
Firstly, we notice that $t$ when fixed as a function of $z$, can admit highest degree of 2 in $z$ and lowest degree 1. Any other choice will change the rank of the singularity at infinity at $t=\infty$.

Let us analyze the meaning of this choice for $t$ in terms of $z$, from the point of view of the associated spectral curve which is of the form \cite{Bonelli:2016qwg}
\begin{equation}
\begin{split}
& y^2 = \frac12 \textrm{Tr} \,\mathbf{A}_I^2
= 4z^3 + 2tz + 2 H_I, 
\\ & 
\sigma_I (q,p) = \frac{p^2}{2} - 2q^3-q t,
\end{split}
\end{equation}
where $\sigma_I(q,p)$ is the Hamiltonian for the Painlev\'e I equation.

Now observe that $(y,z) = (p,q)$ is a point on the spectral curve. For the specific choice of $t$ the spectral curve takes the form 
\begin{equation}
\begin{split}
y^2 & = 8z^3 - 4q^2 z -4q^3 + p^2 
\\ & 
= 8 (z-q)^2 (z+2q) + p^2 -20q^2(z-q).
\end{split}
\end{equation} 

On the point $(y,z) = (p,q)$ then the curve develops a nodal singularity if $p$ hits the value, such that $p^2 = 20q^2(z-q)$, where of course, the fixing of $t$ as a function of $z$ depends on the specific gauge choice for the Lax pair. 
}

\paragraph{Painlev\'e $\textrm{II}_{\textrm{JM}}$.} Let us consider the Painlev\'e $\textrm{II}_{\textrm{JM}}$ example from before. We can decompose
\begin{equation} \label{eq:AIIJMq}
    \mathbf{A}_{\textrm{II}_{\textrm{JM}}} = (z^2 + p + t/2) H + u(z-q) X - \frac{2}{u} (p z + \theta + p q) Y.
\end{equation}
We claim that upon quantization, this is equivalent to the KZ connection with one degree $3$ irregular operator at infinity and one regular operator in the fundamental representation of $\slTwo$ at $z$. Upon fixing $l=0$ and $l'=3$ in \eqref{eq:Omegai2} and identifying $z_j$ with $\infty$ and $z_i$ with $z$, equation \eqref{eq:nablaKZ} becomes
\begin{equation} \label{eq:Omegainf}
     \nabla_z = \kappa \frac{\partial}{\partial z} - \sum_{t=1}^3 \Omega_{\infty}^{(0,t)} z^{t-1} = \kappa \frac{\partial}{\partial z} - \Omega_z.
\end{equation}
Identifying the quantization of $\mathbf{A}_{\textrm{II}_{\textrm{JM}}}$ with $\Omega_z$, we see that (setting $t=0$):
\begin{align}
    H_{\infty}^{(1)} = 2 \hat{p}, &\quad H_{\infty}^{(2)} = 0, \quad H_{\infty}^{(3)} = 2, \\
    X_{\infty}^{(1)} = -\frac{2}{u}(\theta + \hat{p}\hat{q}), &\quad X_{\infty}^{(2)} = - \frac{2}{u} \hat{p}, \quad X_{\infty}^{(3)} = 0, \\
    Y_{\infty}^{(1)} = - u \hat{q}, &\quad Y_{\infty}^{(2)} = u, \quad Y_{\infty}^{(3)} = 0.
\end{align}
Then a straightforward calculation shows that using the quantization condition \eqref{eq:quant}, the following relations hold:\footnote{When using the quantization condition $[\hat{p},\hat{q}] = \kappa$, one should rescale $z$ in $\partial_z \Psi = \mathbf{A}(z)\Psi$ by $z \mapsto z/\kappa$ to obtain the KZ equation.}
\begin{align}
    [H^{(1)}_{\infty}, X^{(1)}_{\infty}] = 2 X^{(2)}_{\infty}, &\quad [H^{(1)}_{\infty},X^{(2)}_{\infty}] = 2 X^{(3)}_{\infty}, \\
    [H^{(1)}_{\infty},Y^{(1)}_{\infty}] = -2 Y^{(2)}_{\infty}, &\quad
    [H^{(1)}_{\infty}, Y^{(2)}_{\infty}] = - 2 Y^{(3)}_{\infty}, \\
    [X^{(2)}_{\infty}, Y^{(1)}_{\infty}] &= H^{(3)}_{\infty}.
\end{align}
However, further computation shows that
\begin{align}\label{faillure II}
    [X^{(1)}_{\infty}, Y^{(1)}_{\infty}] &= 2 \hat q  \neq H^{(2)}_{\infty} \\
    [X^{(1)}_{\infty}, Y^{(2)}_{\infty}] &= 0  \neq H^{(3)}_{\infty}.
\end{align}
This can be cured by adding a deformation $H_{\infty}^{(2)} = 2 \hat q$ via the addition of the term $q z H$ to $\mathbf{A}$ and truncating the loop algebra as follows:
\begin{equation}
    X_{\infty}^{(n)} = Y_{\infty}^{(n)} = H_{\infty}^{(n)} = 0 \quad \textrm{for} \quad n \geq 3.
\end{equation}
The deformation term $H_{\infty}^{(2)} = 2 \hat{q}$ can be again obtained by choosing an appropriate $t$:
\begin{equation}\label{cure II}
    t = 2qz.
\end{equation}
Alternatively, we can follow the procedure in \cite{Jimbo_2008}, where $u$ and $\theta$ were treated as quantum operators, and impose the commutation relation
\begin{equation} \label{eq:uthetabr}
    [\hat \theta, \hat u] = - \kappa \hat u.
\end{equation}
As we will see in section \ref{sec:integrability}, this makes sure that the flatness condition \eqref{eq:compatibility} is preserved at the quantum level. Using \eqref{eq:uthetabr}, one can readily check that 
\begin{equation}
    [X_{\infty}^{(1)}, Y_{\infty}^{(1)}] = 0 = H^{(2)}_{\infty}, \quad [X_{\infty}^{(1)}, Y_{\infty}^{(2)}] = 2 = H^{(3)}_{\infty}.
\end{equation}
\footnote{{{Just like in the previous case, one can compute the spectral curve in this case \cite{Bonelli:2016qwg} and check that $(y,z) = (p+q^2+t/2,q)$ is a point on the spectral curve. Setting $t=2qz$, one can rewrite the curve in the nodal form as 
\begin{equation}
y^2 = (z-q)^2(z^2+ 4qz + 8q^2) + (p+2q^2)^2 + (z-q)(12q^3 + 2pq - 2\theta),
\end{equation}
where the residual pieces vanish for $p=-2q^2$ and $\theta = 4q^3$.  }}}

\paragraph{Painlev\'e $\textrm{IV}$.} The corresponding connection $\mathbf{A}$ here is given by
\begin{equation}\label{A_PIV}
    \mathbf{A} = \frac{A^{(0)}}{z} + A^{(1)} + A^{(2)} z = \left(\begin{array}{cc}
        z+t + \frac{1}{z}(\theta_0 - pq) & u(1 - \frac{q}{2z}) \\
        \frac{2}{u}(pq - \theta_0 - \theta_{\infty}) + \frac{2p}{uz}(pq - 2\theta_0) & -z-t - \frac{1}{z}(\theta_0 - pq)
    \end{array}\right).
\end{equation}
From this we deduce
\begin{align}
    A^{(0)} &= H \otimes (\theta_0 - pq) + X \otimes \left(-\frac{1}{2}{u q}\right) + Y \otimes \frac{2p}{u}(pq - 2\theta_0)  \nonumber \\
    A^{(1)} &= H \otimes t + X \otimes u + Y \otimes \frac{2}{u}(pq - \theta_0 - \theta_{\infty})  \nonumber \\
    A^{(2)} &= H~.
\end{align}
This signals the presence of one irregular operator of degree $2$ at infinity and one regular operator at $0$. There is again a degenerate operator at $z$ and moving it closer to $0$ or infinity gives rise to the regular and irregular singularities of the connection. To compare this to the KZ connection, one can use equations \eqref{eq:Omegaij} and \eqref{eq:Omegainf} to read off ($t=0$):
\begin{align}
    H_0^{(0)} = 2(\theta_0 - \hat{p}\hat{q}), &\quad H^{(1)}_{\infty} = 0, \quad H^{(2)}_{\infty} = 2, \\
    X_0^{(0)} = \frac{2\hat{p}}{u}(\hat{p}\hat{q} -2 \theta_0), &\quad X_{\infty}^{(1)} = \frac{2}{u} (\hat{p}\hat{q} - \theta_0 - \theta_{\infty}), \quad X_{\infty}^{(2)} = 0, \\
    Y_0^{(0)} = - \frac{u}{2}\hat{q}, &\quad Y_{\infty}^{(1)} = u, \quad Y_{\infty}^{(2)} = 0.
\end{align}
Again, one can easily check:
\begin{align}
    [H^{(0)}_0, X^{(0)}_0] = 2 X^{(0)}_0, &\quad [H^{(0)}_0, Y^{(0)}_0] = - 2 Y^{(0)}_0, \\
    [H^{(1)}_{\infty}, X^{(1)}_{\infty}] = 2 X^{(2)}_{\infty} &\quad [H^{(1)}_{\infty}, Y^{(1)}_{\infty}] = -2 Y^{(2)}_{\infty} \\
    [X_0^{(0)}, Y_0^{(0)}] = H_0^{(0)}.
\end{align}
However, the following commutators fail to close:
\begin{equation}
   [X^{(1)}_{\infty}, Y^{(1)}_{\infty}] = 0  \neq H^{(2)}_{\infty}.
\end{equation}
Similarly as before, this can be cured by truncating the loop algebra as follows:
\begin{equation}
    X_{\infty}^{(n)} = Y_{\infty}^{(n)} = H_{\infty}^{(n)} = 0 \quad \textrm{for} \quad n \geq 2,
\end{equation}
or equivalently by choosing $t = -z$.

 { Let us introduce the other Lax matrix \cite{Bonelli:2016qwg}
\begin{equation}
\mathbf{B} = z H + u X + (2/u) (pq-\theta_0-\theta_\infty) Y
\end{equation}
From the compatibility condition, one can compute the equations of motion for $(u,p,q)$, among which 
\begin{equation}
\dot{u} = -(2t+q) u
\end{equation}
Looking at the Lax matrix $\mathbf{A}$, we notice that in this gauge the apparent singularity is at $q =2z$. Thus for the choice of $t=-z$, one gets $\dot{u} = 0$, on this slice and hence fixing $t$ in this way is enough to ensure closure of loop algebra. }

Alternatively, following \cite{Jimbo_2008}, we can impose the operator commutation relation
\begin{equation}
    [\hat \theta_{\infty}, \hat u] = - \kappa \hat u,
\end{equation}
giving
\begin{equation}
    [X_{\infty}^{(1)}, Y_{\infty}^{(1)}] = 2 = H_{\infty}^{(2)}.
\end{equation}

\subsection{Whittaker restriction and descent of flatness}
\label{sec:whittaker_descent}

The truncations encountered in the previous subsection, connecting the quantum Painlev\'e connections to KZ equations, can be understood within the framework of \textit{Whittaker modules} \cite{Kostant1978,Kostant1979}. 

To define a \textit{Whittaker vector} $|W\rangle$, one chooses a triangular decomposition of $\widehat{sl_2} = \mathfrak{n}^- \oplus \mathfrak{h} \oplus \mathfrak{n}^+$, where $\mathfrak{h}$ is a Cartan subalgebra, and $\mathfrak{n}^+$ and $\mathfrak{n}^-$ are nilpotent subalgebras consisting of positive and negative roots spaces. Using a Lie algebra homomorphism (or character) $\chi : \mathfrak{n}^+ \rightarrow \mathbb{C}$, one says that $|W\rangle$ is a Whittaker vector of type $\chi$ if the following holds:
\begin{equation} \label{eq:Wcond}
    x \cdot |W\rangle = \chi(x) |W\rangle, \quad \textrm{for all } \quad x \in \mathfrak{n}^+. 
\end{equation}
A simple example is to take the nilpotent subalgebra generated by $X^{(0)}$ and $Y^{(1)}$ as $\mathfrak{n}^+$. Note that all other positive modes are generated via commutators of these. The negative simple root generators are their counterparts, $Y^{(0)}$ and $X^{(-1)}$. Then a Whittaker vector is defined as the simultaneous eigenvector satisfying the Whittaker conditions:
\begin{align}
    X^{(0)} |W\rangle &= \mu_1 |W\rangle, \\
    Y^{(1)} |W\rangle &= \mu_0 |W\rangle,
\end{align}
with $\mu_0$, $\mu_1$ being two complex numbers. Then all other positive modes which arise as commutators of these vanish automatically. Alternatively, we can build a Whittaker module by imposing \textit{Whittaker restrictions} as follows. Given a $\mathfrak{g}$-module $M$ and a non-singular character $\chi$ as above, one defines $I_W$ as the ideal generated by
\begin{equation} 
    I_W = \{(x - \chi(x)) |x \in \mathfrak{n}^+\}.
\end{equation}
The quotient $M/(I_W M)$ then defines the Whittaker module.

These Whittaker modules can be generalized to the irregular KZ cases discussed in section \ref{sec:PainleveKZmap}. To this end one imposes the condition \eqref{eq:Wcond} only for modes of sufficiently high degree. For example, in the case of Painlev\'e I, one imposes
\begin{align}
    X^{(2)} |W\rangle &= \mu_0 |W\rangle = 4 |W\rangle \\
    Y^{(3)} |W\rangle &= \mu_1 |W\rangle = |W\rangle~.
\end{align}
Since $\mu_0=4$ and $\mu_1=1$ are commuting numbers and all higher modes $n \geq 4$ are generated via commutators of $X^{(2)}$ and $Y^{(3)}$ with other positive modes, we see that 
\begin{equation}
    X^{(n)}|W\rangle = Y^{(n)} |W\rangle = H^{(n)} |W\rangle = 0 \quad \textrm{for} \quad n \geq 4.
\end{equation}

Let $\mathcal H_{\mathrm{irr}}$ be the deformed irregular KZ module of Section~3,
equipped with a flat pair $( \nabla_z, \nabla_t)$, and let $I_{\mathrm W}$
denote the left ideal generated by the Whittaker constraints defining the
reduced sector. Write $\mathcal H_{\mathrm{red}}:=\mathcal H_{\mathrm{irr}}/
I_{\mathrm W}\mathcal H_{\mathrm{irr}}$.

\begin{lemma}[Normalizer criterion for restriction]
\label{lem:normalizer_restriction}
Let $ \nabla=\kappa d-\Omega$ be a flat connection on $\mathcal H_{\mathrm{irr}}$
with coefficients in an associative algebra $\mathcal U$ acting on
$\mathcal H_{\mathrm{irr}}$ (in the KZ setting, $\mathcal U$ is an enveloping
algebra built from affine modes). If every coefficient of $\Omega$ lies in the
normalizer
\begin{equation}
\mathrm N(I_{\mathrm W})
:=
\{a\in\mathcal U:\ a\,I_{\mathrm W}\subset I_{\mathrm W}\},
\end{equation}
then $ \nabla$ induces a well-defined flat connection on the quotient
$\mathcal H_{\mathrm{Toda}}$.
\end{lemma}

\begin{proof}
If $v\equiv v'$ modulo $I_{\mathrm W}\mathcal H_{\mathrm{irr}}$, then
$v-v'=aw$ with $a\in I_{\mathrm W}$. For any coefficient $\omega$ of $\Omega$,
normalizer membership implies $\omega a\in I_{\mathrm W}$, hence
$\omega(v-v')\in I_{\mathrm W}\mathcal H_{\mathrm{irr}}$. Therefore $\nabla$
descends. Flatness is an operator identity on $\mathcal H_{\mathrm{irr}}$, so
it remains true after passing to the quotient.
\end{proof}

\begin{remark}
In the concrete irregular KZ realizations of Section~3, each coefficient is
built from affine modes in a graded manner. Commutators with the generators of
$I_{\mathrm W}$ therefore stay inside $I_{\mathrm W}$ by the same graded
relations, which is precisely the normalizer mechanism of
Lemma~\ref{lem:normalizer_restriction}. This is the conceptual reason why the
Whittaker/Toda restriction preserves flatness.
\end{remark}


\section{Integrability}
\label{sec:integrability}

In this section we bridge the quantization procedure to integrability, examining the behavior of operator flat connections after Whittaker reduction. We demonstrate that the reduced connections take values in $\mathfrak{gl}_2(\mathcal{A})$ over the noncommutative Weyl algebra, analyze the class-preserving deformations of Painlev\'e systems and establish their significance for quantum flatness and BPZ equations.

\subsection{Operator Flat Connections over the Weyl Algebra}
\label{sec:op_flat_weyl}

After Whittaker reduction the affine algebra collapses to the Weyl pair
$(p,q)$, and the reduced flat connection takes values in $\mathfrak{gl}_2(W_\kappa) = \mathfrak{gl}_2 \otimes W_\kappa$ where 
\begin{equation} \label{eq:Weyl}
    W_\kappa=\mathbb C\langle p,q\rangle/(pq-qp-\kappa)
\end{equation}
is the Weyl algebra. A formulation inside $\mathfrak{sl}_2\otimes \mathcal{A}$ is not intrinsically
closed when the algebra $\mathcal{A}$ is noncommutative.

\begin{lemma}
There is no natural Lie algebra structure on $\mathfrak{sl}_2\otimes W_\kappa$
defined by
\begin{equation}
[g\otimes a,\; g'\otimes b] = [g,g']\otimes ab
\end{equation}
if $W_\kappa$ is noncommutative.
\end{lemma}

\begin{proof}
Antisymmetry would require $[g,g']\otimes ab=[g,g']\otimes ba$, hence $ab=ba$
for all $a,b$, which fails in $W_\kappa$.
\end{proof}

\begin{remark}[Explicit failure of antisymmetry]
Take $g=X$, $g'=Y$ and $a=p$, $b=q$. Since $[X,Y]=H$,
\begin{equation}
[X\otimes p,\;Y\otimes q]=H\otimes pq,
\qquad
[Y\otimes q,\;X\otimes p]=(-H)\otimes qp.
\end{equation}
Thus
\begin{equation}
[X\otimes p,\;Y\otimes q]+[Y\otimes q,\;X\otimes p]
=
H\otimes(pq-qp)
=
\kappa\,H\otimes 1 \neq 0,
\end{equation}
so the naive tensor bracket cannot define a Lie algebra when $[p,q]=\kappa$.
\end{remark}

Consequently the minimal intrinsically closed setting for operator fields is
$\mathfrak{gl}_2(W_\kappa)$ equipped with the commutator bracket.
Let $D_z$ and $D_t$ denote the covariant operators defined in Section~3,
now interpreted in $\mathfrak{gl}_2(W_\kappa)$.
Flatness is the condition $[D_z,D_t]=0$, which is equivalent to
\eqref{eq:compatibility}.
Since the derivation uses only associativity and the Jacobi identity, it
remains valid for arbitrary associative coefficient algebras.

\begin{lemma}[Forced identity component under commutator]
\label{lem:forced_identity_component}
Embed $\mathfrak{sl}_2$ into $\mathrm{Mat}_2(\mathbb C)$ with generators
$H,X,Y$ and view $\mathfrak{sl}_2\otimes W_\kappa$ inside $\mathfrak{gl}_2(W_\kappa)$.
For $A=bX$ and $B=cY$ with $b,c\in W_\kappa$ one has
\begin{equation}
[A,B]
=
\frac12\Bigl((bc-cb)\,\mathrm{Id}_2 + (bc+cb)\,H\Bigr).
\label{eq:forced_identity_commutator}
\end{equation}
In particular, if $[b,c] \neq 0$ then $[A,B]$ has an $\mathrm{Id}_2$ component
and does not lie in $\mathfrak{sl}_2\otimes W_\kappa$.
\end{lemma}

\begin{proof}
Use $XY=\frac12(\mathrm{Id}_2+H)$ and $YX=\frac12(\mathrm{Id}_2-H)$:
\begin{equation}
[A,B]=bc\,XY-cb\,YX
=
\frac12\Bigl((bc-cb)\mathrm{Id}_2+(bc+cb)H\Bigr).
\end{equation}
\end{proof}

\begin{remark}[Minimal illustrative example]
With $b=p$ and $c=q$, the Weyl relation gives $pq-qp=\kappa$, hence
\begin{equation}
[pX,\;qY]
=
\frac{\kappa}{2}\,\mathrm{Id}_2+\frac12(pq+qp)\,H,
\end{equation}
so the identity component is algebraically forced.
\end{remark}

\begin{remark}[Compatibility with the enveloping-algebra viewpoint]
In the KZ framework, coefficients are in fact valued in the split Casimirs of the enveloping algebra $U(\mathfrak{sl}_2) \otimes U(\mathfrak{sl}_2)$, and flatness is formulated inside this associative algebra. In the present setting, noncommutativity sits in the Weyl algebra $W_\kappa$ itself. The space $\mathfrak{gl}_2(W_\kappa) \cong M_2(\mathbb{C}) \otimes W_\kappa$ naturally accommodates the tensor product of two $\mathfrak{sl}_2$ representations: the fundamental $2\times 2$ matrix representation and a realization via $W_\kappa$ generators. This yields an algebra homomorphism
\begin{equation}
U(\mathfrak{sl}_2) \otimes U(\mathfrak{sl}_2) \to M_2(\mathbb{C}) \otimes W_\kappa \cong \mathfrak{gl}_2(W_\kappa).
\end{equation}
Under this map, the abstract split Casimirs are evaluated as explicit elements of $\mathfrak{gl}_2(W_\kappa)$. The standard matrix commutator in $\mathfrak{gl}_2(W_\kappa)$ now serves as the bracket relevant for flatness, holonomy, and gauge covariance.
\end{remark}

\begin{proposition}
Let $g(z,t)$ be invertible in $\mathfrak{gl}_2(W_\kappa)$ and define
\begin{equation}
D_z^{(g)} = g D_z g^{-1},
\qquad
D_t^{(g)} = g D_t g^{-1}.
\end{equation}
Then flatness is preserved.
\end{proposition}

\begin{proof}
Conjugation preserves commutators, hence $[D_z,D_t]=0$ implies
$[D_z^{(g)},D_t^{(g)}]=0$.
\end{proof}

\begin{remark}[Where Lie-level suffices and where associativity is forced]
The Lie/enveloping-algebra layer governs the KZ modes and the Whittaker
restriction (Lemma~\ref{lem:normalizer_restriction}). By contrast, the first
appearance of matrix commutators, holonomy (ordered exponentials), or gauge
conjugation forces the associative enlargement: Lemma~\ref{lem:forced_identity_component}
shows that the traceless tensor sector is not closed under commutator once
$W_\kappa$ is noncommutative.
\end{remark}

We would also like to remark that two distinct flat operator pairs with distinct monodromies can share an identical Casimir projection. This phenomenon is explained in more detail in Appendix \ref{app:same_casimir_pairs}. As the Casimir appears in the quantum curve/BPZ equation, this shows that the quantum flat connection retains more information than the corresponding quantum curve equation.

\subsection{Deformation theory, isomonodromy, and \texorpdfstring{$\tau$}{tau}--functions}
\label{sec:deformation_isom_tau}

We work over an associative unital $\mathbb{C}$--algebra $\A$ and consider a rank--$2$
operator connection with values in $\mathfrak{gl}_2(\A)$,
\begin{equation}
D=\kappa\,d-\bigl(A(z,t)\,dz+B(z,t)\,dt\bigr),
\qquad
A,B\in \mathfrak{gl}_2(\A).
\label{eq:operator_connection_short}
\end{equation}
Flatness is the zero--curvature condition
\begin{equation}
\partial_tA-\partial_zB+\frac{1}{\kappa}[A,B]=0,
\label{eq:flatness_short}
\end{equation}
equivalently $[D_z,D_t]=0$ with
\begin{equation}\label{flat paire}
D_z=\kappa\partial_z-A,
\qquad
D_t=\kappa\partial_t-B.
\end{equation}

An infinitesimal gauge parameter $\epsilon(z,t)\in \mathfrak{gl}_2(\A)$ acts by
\begin{equation}
\delta_\epsilon A=[A,\epsilon]-\kappa\,\partial_z\epsilon,
\qquad
\delta_\epsilon B=[B,\epsilon]-\kappa\,\partial_t\epsilon.
\label{eq:gauge_short}
\end{equation}
With the Atiyah--Bott symplectic form
\begin{equation}
\Omega=\int \Tr(\delta A\wedge \delta B),
\end{equation}
flatness \eqref{eq:flatness_short} is the corresponding moment--map equation
\cite{AtiyahBott1983}.

We also recall the $z$--parallel transport
\begin{equation}
\kappa\,\partial_zU=A(z,t)\,U,
\qquad
U(z_0,t)=\mathbf 1,
\label{eq:parallel_transport}
\end{equation}
defined locally on simply connected $z$--domains.

\medskip

\begin{theorem}[Local solvability and ambiguity for fixed $A$]
\label{thm:solveB_cmp}
Let $A(z,t)$ be smooth (or formal) in $t$ and defined on a simply connected
domain in $z$.

\begin{enumerate}[label=\textup{(\roman*)}, leftmargin=*, itemsep=0.3em]
\item There exists locally at least one $B(z,t)$ solving
\eqref{eq:flatness_short}.

\item If $B$ is a solution, then so is $B+C$ for any
$C$ satisfying $[D_z,C]=0$.

\item If $B$ is required to lie in a prescribed singularity
class (for instance polynomial in $z$ of bounded degree, or meromorphic
with fixed pole orders), then solvability imposes compatibility constraints
on the $t$--dependence of $A$.
\end{enumerate}
\end{theorem}

\begin{proof}
Fix $t$ and let $U$ be the $z$--parallel transport
\eqref{eq:parallel_transport}. Conjugation by $U$
reduces the compatibility equation to a first--order linear equation in $z$,
locally solvable on a simply connected domain.
Conjugating back yields $B$, proving (i).
If $B_1,B_2$ solve \eqref{eq:flatness_short},
their difference $C=B_1-B_2$ satisfies $[D_z,C]=0$,
and conversely $B\mapsto B+C$ preserves flatness.
Finally, the construction of $B$ involves conjugation and $z$--integration;
requiring these operations to preserve a given singularity class
produces the usual isomonodromic compatibility constraints.
\end{proof}

\medskip

We now describe the deformation--theoretic tangent space.
For an infinitesimal variation
\begin{equation}
A\mapsto A+\varepsilon\,\alpha,
\qquad
B\mapsto B+\varepsilon\,\beta,
\qquad
\varepsilon^2=0,
\label{eq:inf_var_FT}
\end{equation}
flatness linearizes to
\begin{equation}
[D_z,\beta]-[D_t,\alpha]=0.
\label{eq:FT_linearized}
\end{equation}

\begin{proposition}[Deformation complex]
\label{prop:FT_cohomology}
Infinitesimal gauge transformations generated by $\gamma$ act by
\begin{equation}
\alpha=[D_z,\gamma],
\qquad
\beta=[D_t,\gamma].
\label{eq:FT_gauge}
\end{equation}
Hence the tangent space to the moduli of flat operator connections modulo gauge
is the first cohomology of the two--term complex determined by
\eqref{eq:FT_linearized} modulo \eqref{eq:FT_gauge}.
If $A$ is fixed, flat--preserving deformations of $B$
are precisely those $\beta$ satisfying $[D_z,\beta]=0$.
\end{proposition}

\begin{proof}
Expand $[D_z-\varepsilon\alpha,D_t-\varepsilon\beta]$
to first order to obtain \eqref{eq:FT_linearized}.
For gauge, write $g=\mathbf 1+\varepsilon\gamma$
and compute $gD_\mu g^{-1}
=D_\mu-\varepsilon[D_\mu,\gamma]$,
which gives \eqref{eq:FT_gauge}.
\end{proof}

If $U$ solves \eqref{eq:parallel_transport}, then for any constant
matrix $C_0$ one has $C=U C_0 U^{-1}$ and $[D_z,C]=0$.
Thus $B+\varepsilon C$ is again flat.
These commuting directions are invisible to scalarizations
that retain only quadratic spectral invariants such as
$\Tr(A^2)$.

\medskip

Let $\Psi$ satisfy $D_z\Psi=0$ and choose a formal gauge at $z=\infty$,
\begin{equation} 
\Psi=\widehat G\,e^{\Theta/\kappa},
\qquad
\widehat G=\mathbf 1+\frac{G_1}{z}+\frac{G_2}{z^2}+\cdots .
\label{eq:formal_sol}
\end{equation}

\begin{theorem}[JMU $1$--form]
\label{thm:JMU_closedness}
Define
\begin{equation}
d\log\tau_{\mathrm{JMU}}
=
-\mathrm{res}_{z=\infty}
\Tr\!\left(
\widehat G^{-1} d\widehat G\,\partial_z\Theta
\right).
\label{eq:JMU_1form_general}
\end{equation}
If \eqref{eq:flatness_short} holds,
then $d(d\log\tau_{\mathrm{JMU}})=0$.
\end{theorem}

\begin{proof}
Let $\Omega=\kappa^{-1}(A\,dz+B\,dt)$.
Flatness implies $d\Omega-\Omega\wedge\Omega=0$.
After gauge transformation by $\widehat G$
and using the formal decomposition
\eqref{eq:formal_sol},
the exterior derivative of
\eqref{eq:JMU_1form_general}
reduces to residues of total $z$--derivatives
and traces of commutators, both of which vanish.
\end{proof}

\paragraph{Deformations away from isomonodromic leaves.}
The Painlev\'e Lax pairs can depend, apart from the time variable $t$, on one or more deformation variables $\theta_\alpha$. These parametrize local monodromies near points $z_{\alpha}$ on the $z$-plane when taking the limit $z \rightarrow z_{\alpha}$. In fact, the Painlev\'e Lax pairs are parametrized in such a way that one can find gauge transformations $\widehat G$ with
\begin{equation} \label{eq:gaugetrf}
    \widehat{G} = \mathbf{1} + (z-z_{\alpha}) G_1 + (z-z_\alpha)^2 G_2 + \cdots,
\end{equation}
such that for $z_\alpha$ an irregular singularity of rank $r$, the wavefunction takes the form
\begin{equation}
    \Psi = \widehat{G} e^{\Theta_\alpha/\kappa} e^{\sum_{i=1}^{r}\frac{\Lambda_{i-1}}{i!} (z-z_\alpha)^{-i}}, \quad \Theta_\alpha = \theta_\alpha H \log (z-z_\alpha)~.
\end{equation}
Hence, one can read off the local monodromy to be $\pm \theta_\alpha$. For example, for Painlev\'e IV, one can find a gauge transformation of the form \eqref{eq:formal_sol} with
\begin{equation}
    \Theta = - \theta_\infty H \log z,
\end{equation}
which shows that the local monodromy around infinity is $\pm \theta_{\infty}$. We see that if we allow $\theta_\alpha$ to vary, we are moving \textit{transversely} to the isomonodromic leaves. The corresponding connection 1-form that lives in a three-dimensional space with coordinates $(z,t,\theta_\alpha)$ is now described by 
\begin{equation}
    \Omega = A dz + B dt + C^\alpha d\theta_{\alpha},
\end{equation}
where the matrices $A$, $B$, and $C$ are given by
\begin{equation}
    A = (\partial_z \Psi) \Psi^{-1}, \quad B = (\partial_t \Psi)\Psi^{-1}, \quad C^\alpha = (\partial_{\theta_{\alpha}} \Psi)\Psi^{-1}.
\end{equation}
The flatness condition $d\Omega - \Omega \wedge \Omega$ then translates to
\begin{align}
    \partial_t A - \partial_z B + [A,B] &= 0, \\
    \partial_{\theta_\alpha} A - \partial_z C^\alpha + [A,C^\alpha] &= 0, \\
    \partial_{\theta_\alpha} B - \partial_t C^\alpha + [B,C^\alpha] &= 0.
\end{align}
Defining
\begin{equation}
    \delta A = \partial_z C - [A,C], \quad \delta B = \partial_t C - [B,C],
\end{equation}
we see that these describe infinitesimal deformations of $A$ and $B$ in the framework of equations \eqref{eq:inf_var_FT} and \eqref{eq:FT_linearized}. Note that although the form of $\delta A$ and $\delta B$ is reminiscent to gauge transformations, they are actual deformations as the $C^\alpha$ are not rational in $z$ but rather transcendental and not single valued in $z$.

\subsection{Class-preserving deformations and integrable directions}
\label{sec:PII_deformations}

We now give a complete and explicit realization of the general deformation framework developed in the previous section in the case of a Lax pair system of Poincar\'e rank two as our main example. The resulting flat connection defines a Painlev\'e~II system together with a
finite-dimensional space of class-preserving deformations. By class-preserving we mean here deformations that preserve the polynomial and pole structure in $z$. All other Painlev\'e cases can be treated analogously (see Appendix \ref{sec:PIV-class-preserving} for a treatment of Painlev\'e IV).

\medskip

We start by restricting to the Painlev\'e~II slice by specifying the 
connection coefficients in the form
\begin{equation}
A(z,t)=\Bigl(z^2+p+\frac{t}{2}\Bigr)H+u(z-q)X-\frac{2}{u}(pz+\theta+pq)Y,
\label{eq:PII-A}
\end{equation}
\begin{equation}
B(z,t)
=
\frac{z}{2}H
+
\frac{u}{2}X
-
\frac{p}{u}Y,
\label{eq:PII-B}
\end{equation}
where the deformation variables are parametrized by $\delta \theta$ and $\delta u$. We now evaluate explicitly the flatness condition
\eqref{eq:flatness_short} for the pair \eqref{eq:PII-A}--\eqref{eq:PII-B}.
A direct computation shows that the zero-curvature
equation is equivalent to the system
\begin{equation}
\dot u = -q u,\qquad
\dot q = p + q^2 + \frac{t}{2},\qquad
\dot p = -pq - qp - \theta,\qquad
\dot\theta = 0.
\label{eq:PII-flow}
\end{equation}

\medskip

\begin{proposition}[Hamiltonian structure]
The evolution of $(p,q)$ defined by \eqref{eq:PII-flow} is Hamiltonian
with respect to the Weyl relation \eqref{eq:Weyl}, with Hamiltonian
\begin{equation}
H_{\mathrm{PII}}
=
\frac12 p^2
+
\frac12\left(pq^2+q^2p\right)
+
\frac t2 p
+
\theta q.
\end{equation}

\end{proposition}

\begin{proof}
Using only \eqref{eq:Weyl}, one computes
\begin{equation}
[H_{\mathrm{PII}},q]
=
\frac{1}{2}[p^2,q] + \left(q^2+\frac{t}{2}\right)[p,q]
=
\kappa\left(p+q^2+\frac{t}{2}\right),
\end{equation}
and
\begin{equation}
[H_{\mathrm{PII}},p]
=
-\kappa(pq+qp+\theta).
\end{equation}
Therefore the Heisenberg equations
\begin{equation}
\dot q=\frac{1}{\kappa}[H_{\mathrm{PII}},q],
\qquad
\dot p=\frac{1}{\kappa}[H_{\mathrm{PII}},p]
\end{equation}
reproduce the $q$ and $p$ equations in \eqref{eq:PII-flow}.
\end{proof}

\medskip

We now determine the tangent space of class-preserving infinitesimal
deformations.

\begin{proposition}[Tangent space of class-preserving deformations]
\label{prop:PII-tangent}
Let $A(z,t)$ and $B(z,t)$ be given by
\eqref{eq:PII-A} and \eqref{eq:PII-B}. Then:

\medskip

\noindent
(i) every class-preserving infinitesimal deformation of $A$ is uniquely
parametrized by a quadruple $(\delta p,\delta q,\delta u,\delta\theta)$;

\medskip

\noindent
(ii) there exists a class-preserving companion $\delta B$ if and only if
these parameters satisfy the linearization of \eqref{eq:PII-flow};

\medskip

\noindent
(iii) in that case the companion $\delta B$ is uniquely determined.
\end{proposition}

\begin{proof}
We work in the reduced Painlev\'e~II polynomial class fixed by
\eqref{eq:PII-A} and \eqref{eq:PII-B}. By definition, a class-preserving
infinitesimal deformation is a pair $(\delta A,\delta B)$ such that
$\delta A$ and $\delta B$ remain in the same polynomial ansatz in $z$,
with the same allowed powers of $z$ and the same matrix components as in
\eqref{eq:PII-A} and \eqref{eq:PII-B}.

Differentiating \eqref{eq:PII-A} with respect to the reduced variables
gives
\begin{equation}
\delta A
=
(\delta p)H
+
\bigl((\delta u)z-\delta(uq)\bigr)X
-
2\delta\left(\frac{pz+\theta+pq}{u}\right)Y.
\label{eq:deltaA-parametric}
\end{equation}
Thus a class-preserving variation of $A$ has the form
\begin{equation}
\delta A
=
\alpha H+(\beta_1z+\beta_0)X+(\gamma_1z+\gamma_0)Y.
\end{equation}
Comparing coefficients in \eqref{eq:deltaA-parametric} gives
$\alpha=\delta p$, $\beta_1=\delta u$ and
$\beta_0=-(\delta u)q-u\delta q$, while the two $Y$-coefficients are
uniquely determined by $\delta p,\delta q,\delta u,\delta\theta$ through
the variation of $(pz+\theta+pq)/u$. Hence the class-preserving
variations of $A$ are uniquely parametrized by
$(\delta p,\delta q,\delta u,\delta\theta)$. This proves (i).

For the companion, the class-preserving variation of
\eqref{eq:PII-B} is necessarily
\begin{equation}
\delta B
=
\frac{1}{2}\delta u\,X
-
\delta\left(\frac{p}{u}\right)Y.
\label{eq:deltaB-parametric}
\end{equation}
Since $(A,B)$ is flat, the infinitesimal flatness condition is
\begin{equation}
\partial_t(\delta A)-\partial_z(\delta B)
+\frac{1}{\kappa}[\delta A,B]
+\frac{1}{\kappa}[A,\delta B]
=0.
\label{eq:lin-flatness-P2-proof}
\end{equation}
Substituting \eqref{eq:deltaA-parametric} and
\eqref{eq:deltaB-parametric} into \eqref{eq:lin-flatness-P2-proof}
gives a polynomial identity in $z$. Equating the coefficients of
$H$, $X$ and $Y$ at each power of $z$ yields precisely the
linearization of \eqref{eq:PII-flow}. Explicitly, one obtains
\begin{equation}
\partial_t(\delta u)
=
-(\delta q)u - q\,\delta u,
\end{equation}
\begin{equation}
\partial_t(\delta q)
=
\delta p + q\,\delta q + \delta q\,q,
\end{equation}
and
\begin{equation}
\partial_t(\delta p)
=
-\delta p\,q
-
p\,\delta q
-
\delta q\,p
-
q\,\delta p
-
\delta\theta,
\end{equation}
together with
\begin{equation}
\partial_t(\delta\theta)=0.
\end{equation}
Thus a class-preserving companion exists if and only if
$(\delta p,\delta q,\delta u,\delta\theta)$ satisfies the linearized
Painlev\'e~II system. This proves (ii).

It remains to prove uniqueness. Suppose that $\delta B_1$ and
$\delta B_2$ are two class-preserving companions associated with the
same $\delta A$. Then $\Delta B=\delta B_1-\delta B_2$ satisfies
\begin{equation}
-\partial_z(\Delta B)+\frac{1}{\kappa}[A,\Delta B]=0.
\label{eq:homogeneous-deltaB}
\end{equation}
Since both companions are class-preserving variations of
\eqref{eq:PII-B}, their difference has the form
\begin{equation}
\Delta B=\beta X+\gamma Y,
\end{equation}
with $\beta,\gamma$ independent of $z$. The leading term of $A$ is
$z^2H$. The $z^2$-part of $[A,\Delta B]$ is therefore
\begin{equation}
[z^2H,\beta X+\gamma Y]
=
2\beta z^2X-2\gamma z^2Y.
\end{equation}
Equation \eqref{eq:homogeneous-deltaB} forces $\beta=\gamma=0$.
Hence $\Delta B=0$, so the companion is unique. This proves (iii).
\end{proof}

\medskip

We now identify distinguished integrable directions inside this
deformation space.

An auxiliary flow $\partial_s$ is class-preserving and integrable if
there exists an operator $C(z,t)$ such that
\begin{equation}
\partial_s A - \partial_z C + \frac{1}{\kappa}[A,C]=0,
\qquad
\partial_t C - \partial_s B + \frac{1}{\kappa}[B,C]=0.
\end{equation}

Two natural such directions arise.

\medskip

\noindent
The first is the Painlev\'e time flow itself:
\begin{equation}
X_1(q)=p+q^2+\frac{t}{2},\qquad
X_1(p)=-pq-qp-\theta,\qquad
X_1(u)=-qu,\qquad
X_1(\theta)=0.
\end{equation}

\medskip

\noindent
The second is the residual scaling flow:
\begin{equation}
X_0(q)=0,
\qquad
X_0(p)=0,
\qquad
X_0(u)=u,
\qquad
X_0(\theta)=0,
\end{equation}
with companion
\begin{equation}
C_0=\frac{\kappa}{2}H.
\end{equation}

\begin{proposition}[Distinguished commuting deformation directions]
\label{prop:PII-distinguished}
The reduced Painlev\'e~II system admits two distinguished
class-preserving integrable directions, namely $X_1$ and $X_0$ with
companions $C_1$ and $C_0$. These flows commute and any linear
combination defines a commuting flat deformation.
\end{proposition}

\begin{proof}
For $X_1$, taking $C_1=B$ reduces the auxiliary flatness equations to
the original flatness condition.

For $X_0$, since $C_0=\frac{\kappa}{2}H$, one has
\begin{equation}
\frac{1}{\kappa}[A,C_0]
=
\frac{1}{2}[A,H].
\end{equation}
Using the $\mathfrak{sl}_2$ relations $[H,X]=2X$ and $[H,Y]=-2Y$, we obtain
\begin{equation}
\frac{1}{\kappa}[A,C_0]
=
-\,u(z-q)X
-
\frac{2}{u}(pz+\theta+pq)Y.
\end{equation}
This exactly cancels the variation $X_0(A)$, hence
\[
X_0(A)-\partial_z C_0+\frac{1}{\kappa}[A,C_0]=0.
\]
Since $C_0$ is independent of $z$, the second auxiliary flatness equation
is verified similarly, and the flow $X_0$ commutes with $X_1$.

Linearity implies that any linear combination again defines a flat
commuting deformation.
\end{proof}
\medskip

\begin{proposition}[Rigidity of the deformation variables]
\label{prop:PII-rigidity}
In the reduced Painlev\'e~II system:

\medskip

\noindent
(i) the variable $u$ generates the residual scaling direction;

\medskip

\noindent
(ii) the parameter $\theta$ is central and does not generate an
independent commuting flow.
\end{proposition}

\begin{proof}
Statement (i) follows from the explicit construction of $X_0$.

For (ii), suppose there exists an independent class-preserving deformation $X_\theta$ acting only on $\theta$, such that its induced variations are $X_\theta(\theta) = 1$ and $X_\theta(p) = X_\theta(q) = X_\theta(u) = 0$. By Proposition~\ref{prop:PII-tangent}, for $X_\theta$ to admit a class-preserving companion and thus define a commuting flat deformation, its variations $(\delta p=0, \delta q=0, \delta u=0, \delta\theta=1)$ must satisfy the linearization of the Painlev\'e~II flow \eqref{eq:PII-flow}.

However, substituting these values into the linearized equation for $p$, which reads
\begin{equation}
\partial_t(\delta p)
=
-\delta p\,q
-
p\,\delta q
-
\delta q\,p
-
q\,\delta p
-
\delta\theta,
\end{equation}
we obtain $\partial_t(0) = -1$, which yields the contradiction $0 = -1$. Therefore, a pure shift in $\theta$ violates the compatibility equations, and no such independent commuting flow can exist.
\end{proof}

\medskip

We now show that the two distinguished directions identified above
exhaust all class-preserving integrable flows within the reduced
Painlev\'e~II slice.

\begin{proposition}[Completeness of class-preserving integrable directions]
\label{prop:PII-completeness}
Consider infinitesimal class-preserving flows acting on the reduced
Painlev\'e~II system, i.e. derivations $X$ on the reduced coefficient
algebra such that:

\medskip

\noindent
(i) $X$ preserves the polynomial structure of the reduced operators
$A(z,t)$ and $B(z,t)$, so that the induced variations remain within
the class of deformations parametrized by $(\delta p,\delta q,\delta u,\delta\theta)$;

\medskip

\noindent
(ii) there exists a companion $C(z,t)$ such that the auxiliary flatness
equation holds;

\medskip

\noindent
(iii) $X$ commutes with the Painlev\'e time evolution.

\medskip

Then the space of such derivations is two-dimensional and is spanned by
the vector fields $X_1$ and $X_0$ introduced above. In particular,
every class-preserving integrable direction is a linear combination of
these two flows.
\end{proposition}

\begin{proof}
    See Appendix \ref{A:proof1}.
\end{proof}

We now determine the algebraic structure of the deformation sector.
The key observation is that, among the distinguished commuting flows
identified above, the only one acting nontrivially on the deformation
variables is the multiplicative scaling of $u$. This rigidity strongly
constrains the possible deformation algebras and allows one to classify
all compatible Poisson structures on the two-dimensional sector
generated by $(u,\theta)$ under the requirement that this residual
direction admits a Hamiltonian realization.

\begin{theorem}[Classification of the deformation sector in the reduced PII slice]
\label{thm:PII-classification}
Consider the reduced Painlev\'e~II system defined by
\eqref{eq:PII-A}--\eqref{eq:PII-flow}. Let the deformation sector be the
two-dimensional space generated by $(u,\theta)$ and assume that:

\medskip

\noindent
(i) the deformation variables commute with the canonical sector generated by
$(p,q)$;

\medskip

\noindent
(ii) the residual scaling direction
\begin{equation}
X_0(u)=u,
\qquad
X_0(\theta)=0
\end{equation}
is Hamiltonian with respect to a nondegenerate Poisson structure on the
deformation sector;

\medskip

\noindent
(iii) the Poisson structure is defined on a multiplicative chart $u\neq 0$.

\medskip

\noindent
Then the deformation sector is unique up to reparametrization of $\theta$:
there exists a deformation coordinate, still denoted $\theta$, such that
\begin{equation}
\{\theta,u\}=u.
\label{eq:PII-logcanonical}
\end{equation}
Equivalently, introducing $\phi=\log u$, one has
\begin{equation}
\{\theta,\phi\}=1.
\end{equation}
Upon quantization, this yields uniquely
\begin{equation}
[\theta,u]=\kappa u.
\label{eq:PII-quantum-deformation}
\end{equation}
\end{theorem}

\begin{proof}
Since the deformation sector is two-dimensional, any Poisson structure
on it is determined by a single function
\begin{equation}
J(u,\theta):=\{\theta,u\}.
\end{equation}
For any observable $f(u,\theta)$, the associated Hamiltonian vector
field is therefore
\begin{equation}
X_f
=
J(u,\theta)\left(
\frac{\partial f}{\partial \theta}\frac{\partial}{\partial u}
-
\frac{\partial f}{\partial u}\frac{\partial}{\partial \theta}
\right).
\label{eq:general-Hamiltonian-field}
\end{equation}

By assumption, the residual scaling direction $X_0$ is Hamiltonian, so
there exists a function $H_0(u,\theta)$ such that
\begin{equation}
X_0(f)=\{f,H_0\}
\end{equation}
for all $f$. Applying this identity to the coordinate functions $u$ and
$\theta$ gives
\begin{equation}
\{u,H_0\}=u,
\qquad
\{\theta,H_0\}=0.
\label{eq:H0-system}
\end{equation}

Using \eqref{eq:general-Hamiltonian-field}, these conditions become
\begin{equation}
J(u,\theta)\frac{\partial H_0}{\partial \theta}=u,
\qquad
J(u,\theta)\frac{\partial H_0}{\partial u}=0.
\label{eq:JH0-system}
\end{equation}

Since the Poisson structure is nondegenerate on the chart $u\neq 0$,
one has $J(u,\theta)\neq 0$. The second equation in
\eqref{eq:JH0-system} therefore implies
\begin{equation}
\frac{\partial H_0}{\partial u}=0,
\end{equation}
so that $H_0$ depends only on $\theta$. Writing $H_0=h(\theta)$ and
substituting into the first equation yields
\begin{equation}
J(u,\theta)\,h'(\theta)=u.
\end{equation}
Hence
\begin{equation}
J(u,\theta)=\frac{u}{h'(\theta)}.
\end{equation}

Define a new deformation coordinate by
\begin{equation}
\widetilde{\theta}=\int^\theta h'(\vartheta)\,d\vartheta.
\end{equation}
Then
\begin{equation}
\{\widetilde{\theta},u\}
=
h'(\theta)\,\frac{u}{h'(\theta)}
=
u.
\end{equation}
Thus, up to reparametrization of $\theta$, the Poisson structure is
uniquely of log-canonical form
\eqref{eq:PII-logcanonical}

Introducing $\phi=\log u$, one obtains
\begin{equation}
\{\theta,\phi\}
=
\frac{1}{u}\{\theta,u\}
=
1.
\end{equation}
Quantization replaces Poisson brackets by commutators, yielding \eqref{eq:PII-quantum-deformation}.
\end{proof}

\medskip

This completes the classification of the deformation sector in the reduced Painlev\'e II slice. 

\subsection{Seiberg-Witten geometry and the quantized deformation space}

In the previous section we saw that variations in the parameters $(\theta_\alpha, u_\alpha)$ correspond to deformations of the Lax pair. We would now like to understand the role of these parameters from the point of view of Seiberg-Witten geometry. 

Let us take again the Painlev\'e IV case as an example. In the gauge theory context, the spatial Lax matrix $A(z)$ becomes the Higgs field $\Phi(z)$ of the Hitchin system. The spectral curve of the matrix, $\det(y - A(z))=0$, is precisely the \textit{Seiberg-Witten} curve, and $y dz$ is the Seiberg-Witten differential $\lambda_{\textrm{SW}}$. The parameter $\theta_\infty$ dictates the residue of the connection at the irregular puncture. In Seiberg-Witten geometry, the residues of $\lambda_{SW}$ at the punctures correspond directly to the \textit{mass parameters} of the flavor symmetry group of the 4d $\mathcal{N}=2$ theory determined by the residue at $z=\infty$:
\begin{equation}
    \frac{1}{2\pi i} \oint_{z=\infty} \lambda_{SW} = \textrm{Res}_{z=\infty}(y dz).
\end{equation}
Let us see how this comes about in our case of Painlev\'e IV. Because $A(z)$ is traceless, the equation takes the form:
\begin{equation}
    y^2 + \det(A(z)) = 0 \Rightarrow y = \pm \sqrt{-\det(A(z))}.
\end{equation}
Applying the gauge transformation \eqref{eq:gaugetrf}, one can bring $A(z)$ to the formal form
\begin{equation}
    \tilde A(z) = \Lambda_1 z + \Lambda_0 + \frac{M_\infty}{z}, \quad M_\infty = \theta_\infty H.
\end{equation}
The corresponding eigenvalues are 
\begin{equation}
    y_{\pm}(z) \sim \pm \left(\lambda_1 z + \lambda_0 + \frac{\theta_\infty}{z} + \mathcal{O}(z^{-2})\right).
\end{equation}
Hence we see that when computing the residue at infinity, we get
\begin{equation}
    \textrm{Res}_{z=\infty}(y_{\pm}dz) = \pm \theta_\infty,
\end{equation}
which shows that $\theta_\infty$ is a mass parameter (see also \cite{GamayunIorgovLisovyy2012,GamayunIorgovLisovyy2013} for detailed mappings between monodromy and mass parameters for all Painlev\'e cases). If $\theta_\infty$ is a mass $m$, what is its symplectic conjugate? Note that the conjugate coordinate is the period dual to the mass:
\begin{equation}
    a_{D,mass} = \frac{\partial \mathcal{F}}{\partial m},
\end{equation}
where $\mathcal{F}$ is the prepotential of the theory. Within the gauge theory/integrable systems correspondence, the prepotential is related to the tau function via $\mathcal{F} \sim \log \tau$. One can prove that the gauge parameter $u$ of the Painlev\'e system can be expressed as \cite{JimboMiwa1981_II}
\begin{equation}
    u = \frac{\tau(\theta_\infty + 1)}{\tau(\theta_\infty)},
\end{equation}
which in the continuum limit translates to 
\begin{equation}
    \log u \sim \frac{\partial \log \tau}{\partial \theta_\infty}.
\end{equation}
Thus we conclude
\begin{equation}
    \frac{\partial \mathcal{F}}{\partial m} \sim \frac{\partial \log \tau}{\partial m} = \frac{\partial \log \tau}{\partial \theta_\infty} = \log u.
\end{equation}
Hence, $\log u$ is the symplectic dual to $\theta_\infty$. 

\subsection{Quantum Flat Connections over the algebra $\mathcal{A}$}
\label{sec:painleve_op_flatness}

In this section we want to analyze the flatness condition \eqref{eq:compatibility} in the case when $p$ and $q$ are promoted to operators. In the Heisenberg picture, they will be time dependent with their time derivative given by
\begin{equation} \label{eq:Heisenbergflow}
    \partial_t p = \frac{1}{\kappa}[\hat H,p], \quad \partial_t q = \frac{1}{\kappa}[\hat H,q],
\end{equation}
with a suitable Hamiltonian $\hat H$ such that the flatness condition 
\begin{equation} \label{eq:flatness}
    \partial_t A - \partial_z B + [A,B] = 0
\end{equation}
is satisfied. We find that, apart from the rigid Painlev\'e I system, all other Painlev\'e Lax pairs are mutually quantum flat if and only if one also quantizes the corresponding deformation space. That is, the corresponding flat connections take values in $\mathfrak{gl}_2(\mathcal{A})$ where in general
\begin{equation}
    \mathcal{A} = W_\kappa(p_i,q_i) \otimes W_{-\kappa}(\theta_\alpha,\log(u_\alpha)).
\end{equation}
We demonstrate this explicitly in the cases of Painlev\'e $\textrm{II}_{\textrm{JM}}$ and $\textrm{IV}$.

\paragraph{Quantum Toda Lax Pair.} For autonomous Hamiltonian systems like the Toda chain (or the KdV hierarchy), the Lax matrix $L(z,t)$ defines an \textit{eigenvalue problem}, not a differential equation in $z$. The auxiliary linear system is:
\begin{align}
    L(z,t) \Psi &= \lambda \Psi \\
    \partial_t \Psi &= M(z,t) \Psi.
\end{align}
Notice that there is no $\partial_z \Psi$ and the spectral parameter $z$ is strictly a constant parameter. If we demand that the spectrum remains constant as the system evolves in time $(\partial_t \lambda = 0$), we obtain
\begin{equation}
    (\partial_t L)\Psi + L(\partial_t \Psi) = \lambda (\partial_t \Psi).
\end{equation}
Substituting $\partial_t \Psi = M \Psi$, we then obtain the following zero-curvature condition:
\begin{equation} \label{eq:Todaflat}
    \partial_t L - [M,L] = 0.
\end{equation}
Concretely, the Lax pair for the periodic Toda chain is given by
\begin{align}
    L &= \left(\begin{array}{cc}
        p & e^{q/2} + z^{-1} e^{-q/2} \\
        e^{q/2} + z e^{-q/2} & -p
    \end{array}\right), \\
    M &= \left(\begin{array}{cc}
        0 & -\frac{1}{2} (e^{q/2} - z^{-1} e^{-q/2}) \\
        \frac{1}{2} (e^{q/2} -z e^{-q/2}) & 0
    \end{array}\right).
\end{align}
The Hamiltonian is given by $\hat H = p^2 + e^q + e^{-q}$. Using $[p, e^{\alpha q}] = \kappa \alpha e^{\alpha q}$, the exact quantum equations of motion are
\begin{equation}
    \partial_t p = - e^q + e^{-q}, \quad \partial_t (e^{\pm q/2}) = \frac{1}{2}[p, e^{\pm q/2}]_+,
\end{equation}
where $[A,B]_+ = AB + BA$ denotes the anticommutator. For the top-right entry $(1,2)$, the matrix commutator yields:
\begin{equation}
    [L,M]_{12} = [p,M_{12}]_+ = - \frac{1}{2} [p,e^{q/2}]_+ + \frac{1}{2}z^{-1} [p,e^{-q/2}]_+.
\end{equation}
This exactly mirrors $- \partial_t L_{12}$. For the top-left entry $(1,1)$, the cross-multiplication of the off-diagonal $q$-dependent terms trivially commutes. The result is $e^q - e^{-q}$, perfectly matching $-\partial_t p$. Therefore, we see that the condition \eqref{eq:Todaflat} is also satisfied at the quantum level.

\paragraph{Quantum Painlev\'e I.} Let us next turn to Painlev\'e I. The Lax pair is given by
\begin{equation}\label{eq:PI_Lax_paire}
    A = - p H + (q^2 + z q + z^2 + \frac{t}{2})X + 4(z-q)Y, \quad B = (q + \frac{z}{2})X + 2 Y.
\end{equation}
Promoting $p$ and $q$ to Heisenberg operators as before, we can compute
\begin{equation}
    [A,B] = (6 q^2 + t)H - (pq + q p + zp)X + 4p Y,
\end{equation}
as well as
\begin{equation}
    \partial_t A = -p_t H + (pq + qp + zp + \frac{1}{2})X - 4 q_t Y, \quad \partial_z B = \frac{1}{2} X. 
\end{equation}
Implementing the flatness condition \eqref{eq:flatness} and gathering the coefficients of $H$, $X$, and $Y$, one then finds
\begin{itemize}
    \item $H$ component: $$-p_t + 6q^2 + t = 0 \Rightarrow p_t = 6q^2 + t$$

    \item $Y$ component: $$-4 q_t + 4p = 0 \Rightarrow q_t = p$$

    \item $X$ component: $$\partial_t (q^2) + z(\partial_t q) + \partial_t (t/2) = pq + qp + z p  +\frac{1}{2}$$
\end{itemize}
The $X$ component trivially vanishes when $q_t = p$. The $H$ and $Y$ components yield the exact Hamiltonian equations for PI, returning
\begin{equation}
    q_{tt} = 6q^2 + t.
\end{equation}

\paragraph{Quantum Painlev\'e $\textrm{II}_{\textrm{JM}}$.} Now we proceed to Painlev\'e $\textrm{II}_{\textrm{JM}}$ with the Lax pair given by
\begin{equation}
    A = A_H H + A_X X + A_Y Y, \quad B = B_H H + B_X X + B_Y Y,
\end{equation}
where 
\begin{equation}
    A_H = (z^2 + p + t/2), \quad A_X = u(z-q), \quad A_Y = -2/u(pz + \theta + pq),
\end{equation}
and
\begin{equation}
    B_H = z/2, \quad B_X = u/2, \quad B_Y = -p/u.
\end{equation}
A direct computation shows 
\begin{align}
    [A,B]_H &= qp + pq + \theta, \\
    [A,B]_X &= u(zq + p + t/2), \\
    [A,B]_Y &= 2u^{-1}(p^2 - pqz - \theta z + \frac{t}{2}p), \\
    [A,B]_I &= \frac{1}{2}\left([u(z-q), - \frac{p}{u}] + [-\frac{2}{u}(pz + \theta + pq),\frac{u}{2}]\right).
\end{align}
We see that the last commutator, namely $[A,B]_I$ only vanishes, if we impose
\begin{equation}
    [\theta, u] = - \kappa u,
\end{equation}
which is exactly the quantization condition \eqref{eq:uthetabr}.
Furthermore, by taking $t$-derivatives, one obtains
\begin{equation}
    \partial_t A = (p_t + 1/2)H + (u_t (z -q) - u q_t) X - 2u^{-1} [q(pz + \theta + pq) + p_t z + p_t q + p q_t] Y.
\end{equation}
Plugging this into the flatness condition \eqref{eq:flatness}, we see that 
\begin{equation}
    \partial_t p = - p q - q p - \theta, \quad \partial_t q= p + q^2 + t/2, \quad \partial_t u = -qu
\end{equation}
guarantee flatness in the quantum case. 

\paragraph{Quantum Painlev\'e IV.}

Following the same logic as in the previous cases, one can readily check that quantum flatness holds for the Painlev\'e IV case, once we impose the quantization conditions
\begin{equation}
    [p,q] = \kappa, \quad [\theta_{\infty},u] = - \kappa u,
\end{equation}
while treating $\theta_0$ as a commuting variable. Quantum flatness also uniquely fixes the flow equations for the coordinate and the conjugate momentum variables:
\begin{align}
    q_t &= pq + qp - q^2 - tq + 2\theta_0, \\
    p_t &= -p^2 + pq + qp + tq - 2 \theta_\infty.
\end{align}
These are induced via the quantum Hamiltonian
\begin{equation}
    H_{IV} = \frac{1}{2} (p^2 q + qp^2) - \frac{1}{2} (p q^2 + q^2 p - \frac{t}{2}(pq + qp) + 2\theta_0 p + 2\theta_\infty q,
\end{equation}
upon imposing the Heisenberg equations of motion \eqref{eq:Heisenbergflow}.

\subsection{Gauge transformations and the energy-momentum tensor}
\label{sec:Tgauge}

Here we want to see how the BPZ equation arises from the quantization of the flat connections discussed before. To this end note that our classical linear system was given via
\begin{equation} \label{eq:psiA}
    \partial_z \Psi = A \Psi, 
\end{equation}
where $\Psi = (\psi_1,\psi_2)^T$ and $A$ is the first Lax matrix for the corresponding Painlev\'e system. Evaluating the first row, one gets
\begin{equation}
    \partial_z \psi_1 = A_{11} \psi_1 + A_{12}\psi_2.
\end{equation}
This can be rewritten as a second-order equation for $\psi_1$ as follows. As a first step, we solve for $\psi_2$, giving
\begin{equation}
    \psi_2 = \frac{1}{A_{12}}(\partial_z \psi_1 - A_{11} \psi_1).
\end{equation}
Taking the $z$-derivative of this equation and substituting back into \eqref{eq:psiA}, one obtains a second-order ODE for $\psi_1$ which takes the form
\begin{equation}
\label{eq:d2psi}
\psi_1^{''} - \frac{A_{12}'}{A_{12}}\psi_1' + \bigg(\det A + \frac{A_{12}'A_{11}}{A_{12}} - A_{11}' \bigg) \psi_1 = 0
\end{equation}

Next, we define a new basis vector $v$ such that 
\begin{equation}
    \psi_1 = \sqrt{A_{12}}v.
\end{equation}
Plugging this into \eqref{eq:d2psi}, the first derivative term cancels and we are left with an equation of the form 
\begin{equation}
\label{eq:EM_tensor}
    \partial_z^2 v + T(z) v = 0.
\end{equation}
where 
\begin{equation}
\label{eq:EMtensor}
T(z) = \det A + \frac{A_{12}''}{2A_{12}} - \frac{3(A_{12}')^2}{4A_{12}^2} + \frac{A_{12}'A_{11}}{A_{12}} - A_{11}'
\end{equation}

This can be rewritten as a gauge transformation of the original equation. To this end, define a new gauge-transformed vector as $\Psi^S = (v, (\partial_z v)^T)$ with gauge transformation matrix $S$ ($\Psi^S = S \Psi$) given by
\begin{equation}
    S(z,t) = \left(\begin{array}{cc}
        \frac{1}{\sqrt{A_{12}}} & 0 \\
        \frac{1}{\sqrt{A_{12}}}(A_{11} - \frac{\partial_z A_{12}}{2A_{12}}) & \sqrt{A_{12}}
    \end{array}\right)~.
\end{equation}
By applying the standard gauge transformation formula
\begin{equation}
    A^S = S A S^{-1} + (\partial_z S) S^{-1},
\end{equation}
the connection matrix is forced exactly into the specific form
\begin{equation}
    A^S = \left(\begin{array}{cc}
        0 & 1 \\
        T(z) & 0 
    \end{array}\right)~.
\end{equation}
The corresponding equation $\partial_z \Psi^S = A^S \Psi^S$ is known to describe an \textit{Oper}. The explicit form of $T(z)$ is generated by the transformation is given in \eqref{eq:EMtensor}.

\paragraph{Example : Painlev\'e II}
As an example, we explicitly compute $T(z)$ for Painlev\'e II. Substituting
\begin{equation}
    A_{11} = z^2 + p + t/2, \quad A_{12} = u(z - q),
\end{equation}
into the above, we get 
\begin{equation}
    T(z) = \det(A) + 2z + \frac{z^2 + p +t/2}{z-q} + \frac{3}{4(z-q)^2}~.
\end{equation}
Note that the residue of the pole $z \rightarrow q$ is $q^2 + p + t/2$ and exactly matches the conjugate variable given in equation \eqref{eq:PIIphasespace}. 

\paragraph{Example : Painlev\'e I}
Now we turn to Painlev\'e I system again. Note that
choosing $A_{21}\neq 0$, one can derive a similar second order equation for $\psi_2$. Indeed for Painlev\'e I equation the $A_{21} \propto (z-q)$, where this is the suitable choice. 

\begin{equation}
\psi_2'' - \frac{A_{21}'}{A_{21}} \psi_2' + \bigg(\det A + \frac{A_{21}'A_{22}}{A_{21}} - A_{22}'\bigg) \psi_2 = 0
\end{equation}
After rescaling $\psi_2 = \sqrt{A_{21}} v$, one obtains a Schr\"odinger type equation for $v$ as 
\begin{equation}
\partial^2_z v + T(z) v = 0,
\end{equation}
where 
\begin{equation}
T(z) = \frac{A_{21}''}{2A_{21}} - \frac{3(A_{21}')^2}{4 A_{21}^2}
+ \det A + \frac{A'_{21}A_{22}}{A_{21}} - A_{22}'
\end{equation}
For Painlev\'e I, one has 
\begin{equation}
A_{21} = 4(z-q), \qquad A_{22} = p
\end{equation}
from which one can compute 
\begin{equation}
T(z) = \det A + \frac{p}{z-q} - \frac34 \frac{1}{(z-q)^2}
\end{equation}
which matches exactly with equation 3.5 of \cite{vanSpaendonck:2022kit}. Note again that the residue at the pole $z\rightarrow q$ matches the conjugate phase space variable of the corresponding Seiberg-Witten curve.

We want to interpret this result as the energy momentum tensor corresponding to the insertion of a degenerate operator of type $(1,2)$ at position $z=q$. Indeed, in Liouville theory, degenerate operators of type $(m,n)$ have conformal weight
\begin{equation}
    \Delta_{m,n} = \frac{Q^2}{4} - \frac{1}{4}\left(m b + \frac{n}{b}\right)^2,
\end{equation}
where $Q = \left(b + \frac{1}{b}\right)$. Let's look at the specific degenerate field $\Phi_{1,2}$:
\begin{equation}
    \Delta_{1,2} = \frac{(b + 1/b)^2}{4} - \frac{(b + 2/b)^2}{4} = - \frac{1}{2} - \frac{3}{4b^2}.
\end{equation}
The classical limit $b \rightarrow 0$ thus gives
\begin{equation}
    \Delta_{1,2} \sim -\frac{3}{4b^2}.
\end{equation}
Upon quantization, $[\hat{p},\hat{q}]=\kappa$, the oper equation also gets quantized giving rise to the CFT null-vector equation\footnote{The null-state condition is $(b^{-2} L_{-2} + (L_{-1})^2)\Phi_{1,2}(z) = 0$ and one has $L_{-1}^2 = \partial_z^2$}
\begin{equation}
    \left( \partial_z^2 + \widehat{T}_{cft}(z)\right)\langle \Phi_{1,2}(z) \cdots \rangle = 0.
\end{equation} 

Identifying $\kappa=-b^2$, one can recover $\kappa$ in various terms in the energy-momentum tensors. Namely, we see that $\det (A)$ gets rescaled by $\kappa^{-2}$, compare with \eqref{eq:BPZfull}.

In order to fully recover the quantum corrections to the conformal weight of the degenerate operator in the BPZ equation \eqref{eq:BPZfull}, one must perform the quantum version of the gauge transformation \eqref{eq:gaugetrf}. This quantum gauge transformation is best understood in a doubly quantized setting where both commutators $[y,z] = \epsilon_1$ and $[p,q] = \epsilon_2$ are turned on. In this framework, one can show that for suitable choices of the deformation parameters $\epsilon_1$ and $\epsilon_2$, the exact conformal weight of the degenerate operator is recovered. For more general choices of these parameters, one obtains novel generalizations of the ordinary BPZ equation, a topic which we will elaborate on in forthcoming work.

\acknowledgments{We would like to thank Mohamad Alameddine, Wei Cui, Anton Dzhamay, Alexander Hock, Nicolai Reshetikhin, Mauricio Romo and Fengjun Xu for valuable discussions. SB also acknowledges CERN and University of Geneva for hospitality at the final stages of the work. BH would like to thank IHES in Bures-sur-Yvette  and MPIM Bonn, where part of this work was completed, for hospitality. AL would like to express his gratitude to N. Reshetikhin and S.-T. Yau for organizing the conference \textit{Peaks, Depths, and Currents: Mathematics in Asia}, where this collaboration was initiated.}

\appendix

\section{Two distinct flat operator pairs with identical Casimir projection}
\label{app:same_casimir_pairs}

This appendix investigates a specific algebraic subtlety within the Weyl algebra by presenting two distinct flat operator pairs that share the identical Casimir projection. It explores transported commutant deformations to demonstrate how non-unique operator pairs can emerge despite satisfying identical trace-like invariants in a noncommutative setting.

Throughout this section we work in the Weyl algebra
\begin{equation}
W_\kappa=\mathbb C\langle p,q\rangle/(pq-qp-\kappa),
\end{equation}
and in $\mathcal A=\mathrm{Mat}_2(W_\kappa)$.
For a traceless matrix $A=aH+bX+cY$ one has
\begin{equation}
C(A)=\frac12\mathrm{tr}(A^2)=\frac12 a^2+bc,
\label{eq:Casimir_formula_app}
\end{equation}
where $bc$ denotes the ordered product in $W_\kappa$.

\begin{proposition}[Transported commutant deformation]
\label{prop:commutant_deformation_app}
Let $U(z,t)$ denote the $z$–parallel transport of $D_z$,
and let $C_0\in\mathrm{Mat}_2(\mathbb C)$ be noncentral.
Define
\begin{equation}
C(z,t)=U(z,t)\,C_0\,U(z,t)^{-1}.
\end{equation}
If $(A,B)$ satisfies \eqref{eq:compatibility},
then $(A,B+C)$ also satisfies \eqref{eq:compatibility}.
Both pairs have identical quadratic class $C(A)$,
but distinct operator $t$–holonomy whenever $C_0$ is noncentral.
\end{proposition}

\begin{proof}
Differentiation gives $\partial_zC=\kappa^{-1}[A,C]$,
hence adding $C$ to $B$ does not modify the curvature.
The quadratic class depends only on $A$.
Noncentrality of $C_0$ implies that $C(z,t)$
is noncentral and therefore changes the ordered exponential
of $B$, producing distinct operator transport.
\end{proof}

\begin{proposition}[Distinct $A$ with identical quadratic class]
\label{prop:diffA_sameCasimir_app}
Define constant matrices
\begin{equation}
A_{(1)}=qX+pY,
\qquad
A_{(2)}=(qp)X+1\cdot Y.
\end{equation}
Then
\begin{equation}
C(A_{(1)})=C(A_{(2)})=qp.
\end{equation}
Taking $B_{(1)}=B_{(2)}=0$
gives two flat pairs with identical quadratic class
but distinct $z$–holonomy.
\end{proposition}

\begin{proof}
Formula \eqref{eq:Casimir_formula_app}
gives $C(A_{(1)})=qp$
and $C(A_{(2)})=(qp)\cdot1=qp$.
Flatness follows from $t$–independence.
Since $A_{(1)} eq A_{(2)}$,
the exponentials $\exp(\ell A_{(i)}/\kappa)$ differ,
hence the holonomies differ.
\end{proof}

\begin{proposition}[Distinct flat pairs with identical quadratic class]
\label{prop:diffA_diffB_sameC_flat}
Let
\begin{equation}
A_{(1)}=qX+pY,
\qquad
A_{(2)}=(qp)X+1\cdot Y,
\end{equation}
and define
\begin{equation}
B_{(1)}=A_{(1)},
\qquad
B_{(2)}=A_{(2)}.
\end{equation}
Then
\begin{equation}
[D_z^{(i)},D_t^{(i)}]=0,
\qquad i=1,2,
\end{equation}
and
\begin{equation}
C(A_{(1)})=C(A_{(2)})=qp,
\end{equation}
while $(A_{(1)},B_{(1)}) eq(A_{(2)},B_{(2)})$.
\end{proposition}

\begin{proof}
Flatness follows from constancy and
$[A_{(i)},A_{(i)}]=0$.
Equality of quadratic class follows from
\eqref{eq:Casimir_formula_app}.
Non-equality is immediate from the distinct $Y$–coefficients.
\end{proof}

Together these constructions show that the quadratic projection
\begin{equation}
A\longmapsto \frac12\mathrm{tr}(A^2)
\end{equation}
forgets essential non-abelian information.
It does not determine the full operator-valued flat connection,
even up to exact equality of quadratic data.

\section{Proof of Proposition \ref{prop:PII-completeness}}
\label{A:proof1}

This appendix details the proof of Proposition \ref{prop:PII-completeness}, formally establishing the structure of class-preserving infinitesimal deformations in the Painlev\'e II system. 

\begin{proof}
By Proposition~\ref{prop:PII-tangent}, every class-preserving infinitesimal
deformation is uniquely determined by a quadruple
$(\delta p,\delta q,\delta u,\delta\theta)$. The requirement that the flow $X$ commutes with the time evolution $\partial_t$ means that the variations must satisfy the linearized Painlev\'e~II system \eqref{eq:PII-flow}:
\begin{align}
    \partial_t (\delta q) &= \delta p + 2q \delta q, \label{eq:lin-q} \\
    \partial_t (\delta p) &= -2q \delta p - 2p \delta q - \delta\theta, \label{eq:lin-p} \\
    \partial_t (\delta u) &= -u \delta q - q \delta u, \label{eq:lin-u} \\
    \partial_t (\delta\theta) &= 0. \label{eq:lin-theta}
\end{align}
Furthermore, since $X$ is a derivation on the reduced coefficient algebra, the components $(\delta p, \delta q, \delta u, \delta\theta)$ must be expressible as polynomials in the generators $(p, q, u, \theta)$.

Equation \eqref{eq:lin-theta} immediately implies that $\delta\theta = c$ for some constant $c$. 
We first show that $c = 0$. Suppose $c \neq 0$. Differentiating \eqref{eq:lin-q} with respect to $t$ and using \eqref{eq:lin-p} and the background flow \eqref{eq:PII-flow}, we obtain a second-order linear differential equation for $\delta q$:
\begin{equation}
    \partial_t^2 (\delta q) = \partial_t (\delta p) + 2(\partial_t q)\delta q + 2q \partial_t (\delta q).
\end{equation}
Substituting $\partial_t (\delta p)$ and $\partial_t q$, this reduces to
\begin{equation}
    \partial_t^2 (\delta q) = 2(q^2 - p + t)\delta q - c.
\label{eq:PII-variation-theta}
\end{equation}
We now show that \eqref{eq:PII-variation-theta} admits no polynomial solutions in $(p,q)$ unless $c=0$. First, note that for $c=0$, the equation is the linearization of the Painlev\'e II equation itself, which admits the known polynomial solution $\delta q = \partial_t q = p + q^2 + \frac{t}{2}$, corresponding to the time-translation flow $X_1$.

To isolate the effect of $c$, let $\delta q$ be an arbitrary polynomial solution to \eqref{eq:PII-variation-theta}. We can subtract a suitable multiple of the time translation flow, $\lambda (p + q^2 + \frac{t}{2})$, to strictly lower the degree of the remaining polynomial $P(p,q,t) = \delta q - \lambda (p + q^2 + \frac{t}{2})$, such that $P$ no longer contains the maximal invariant combination of $p+q^2$. Since the homogeneous equation is satisfied by $p+q^2+t/2$, the remainder $P$ must still satisfy the non-homogeneous equation:
\begin{equation}
    \partial_t^2 P = 2(q^2 - p + t)P - c.
\end{equation}

We assign degree 2 to $p$, degree 1 to $q$, and treat $t$ and $c$ as constants of degree 0. If $P$ is identically zero, the equation trivially reduces to $0 = -c$, forcing $c=0$.

Assume for contradiction that $P$ is not identically zero, and let $d \ge 0$ be its degree in $(p,q)$. Because the leading-order action of $\partial_t$ is given by the vector field $v_0 = p \partial_q - 2pq \partial_p$, which annihilates $p+q^2$, any polynomial whose leading homogeneous component is not proportional to powers of $(p+q^2)$ will have its degree correctly incremented by 1 under $\partial_t$. Thus, $\partial_t^2 P$ has degree $d + 2$. On the right-hand side, the leading term is $2(q^2 - p)P$, which also has degree exactly $d + 2$ (the constant $-c$ has degree 0, which is strictly less than $d+2$). Thus, the leading degree $d+2$ terms on both sides must match exactly.

We require the highest degree homogeneous component $P_d$ of $P$ to satisfy:
\begin{equation}
    v_0^2(P_d) = 2(q^2 - p) P_d.
\end{equation}
To solve this, notice that $I = p + q^2$ is an invariant of $v_0$, since $v_0(p+q^2) = -2pq + 2qp = 0$.
We can therefore change variables from $(p,q)$ to $(I,q)$ and treat $I$ as a constant parameter. The vector field becomes $v_0 = (I - q^2)\partial_q$.
The equation for $P_d(q,I)$ is
\begin{equation}
    (I - q^2) \partial_q \left( (I - q^2) \partial_q P_d \right) = 2(2q^2 - I) P_d.
\end{equation}
Since $P_d$ is a polynomial in $(p,q)$, it must also be a polynomial in $q$ with coefficients depending on $I$. Let $n \ge 0$ be the degree of $P_d$ as a polynomial in $q$, so its leading term is $C(I) q^n$ for some non-zero polynomial $C(I)$.
Applying $v_0 = (I - q^2)\partial_q$ to the leading term $C(I)q^n$, the highest power of $q$ comes from $-q^2 \partial_q$, giving $-n C(I) q^{n+1}$.
Applying $v_0$ a second time yields a leading term of $- (n+1) (-n C(I)) q^{n+2} = n(n+1) C(I) q^{n+2}$.
On the right-hand side, the leading term of $2(2q^2 - I) P_d$ is $4 C(I) q^{n+2}$.
Matching the leading coefficients of $q^{n+2}$, we must have
\begin{equation}
    n(n+1) C(I) = 4 C(I) \implies n(n+1) = 4.
\end{equation}
However, the equation $n(n+1) = 4$ has no integer solutions (since for $n=1$, $n(n+1)=2$, and for $n=2$, $n(n+1)=6$).
This contradiction implies that our assumption that $P$ is not identically zero is false.

Therefore, the entire polynomial remainder $P$ must be identically zero. Returning to the full non-homogeneous equation $\partial_t^2 P = 2(q^2 - p + t)P - c$, substituting $P = 0$ yields $0 = -c$, so $c = 0$. Therefore, to preserve the polynomial structure of the algebra, we must strictly have $\delta\theta = c = 0$.

Having shown that $c=0$, and that $P=0$ is the unique solution after subtracting multiples of $X_1$, we see that any valid deformation $\delta q$ must just be a multiple of the time translation flow $\partial_t q$. We now classify the remaining polynomial solutions to the homogeneous linearized system. We use the known flows $X_1$ (time translation) and $X_0$ (scaling). Since $X_1$ corresponds to the time derivative itself, we can subtract a suitable multiple $\lambda X_1$ from our deformation $X$ such that the new derivation $Y = X - \lambda X_1$ has variations satisfying $\delta q|_{Y} = 0$. 

Let us verify that this forces the rest of the variations to align with $X_0$. If $\delta q = 0$, equation \eqref{eq:lin-q} gives:
\begin{equation}
    \delta p = \partial_t (\delta q) - 2q \delta q = 0.
\end{equation}
With $\delta q = 0$ and $\delta p = 0$, equation \eqref{eq:lin-p} is trivially satisfied ($0 = 0$).
Finally, substituting $\delta q = 0$ into equation \eqref{eq:lin-u} leaves:
\begin{equation}
    \partial_t (\delta u) = -q \delta u.
\end{equation}
Since the background variable $u$ satisfies the exact same equation ($\dot{u} = -qu$), we have
\begin{equation}
    \partial_t \left( \frac{\delta u}{u} \right) = \frac{\partial_t (\delta u) u - \delta u \partial_t u}{u^2} = \frac{-q \delta u \cdot u - \delta u (-qu)}{u^2} = 0.
\end{equation}
This implies that $\delta u = \mu u$ for some scalar constant $\mu$.

This resulting variation $(\delta p = 0, \delta q = 0, \delta u = \mu u, \delta\theta = 0)$ is exactly the scaling flow $\mu X_0$. Therefore, the modified flow is $Y = \mu X_0$, which means the original arbitrary class-preserving flow $X$ can be written as:
\begin{equation}
    X = \lambda X_1 + \mu X_0.
\end{equation}
Thus, the space of admissible class-preserving integrable directions is exactly two-dimensional and spanned by $X_1$ and $X_0$.
\end{proof}

\section{Class-preserving deformations in the reduced Painlev\'e IV slice}
\label{sec:PIV-class-preserving}

\paragraph{Convention.}
In this appendix we work in the classical (commutative) setting,
i.e. the coefficients of the Lax pair are treated as commuting
variables. The curvature is computed using the standard
$\mathfrak{sl}_2$ component formulas. Quantum corrections arising
from noncommutativity and the $\mathrm{gl}_2(W_\kappa)$ framework are
treated in Section~4 and are not needed here.

\medskip

We construct a class-preserving Lax representation of the Painlev\'e~IV
system by determining a pair $(A_{IV},B_{IV})$ with fixed singularity
structure whose flatness reproduces the evolution equations.
Throughout, the variables $(p,q)$ are commuting, and the evolution is
governed by
\begin{align}
\partial_t q
&=
pq+qp-q^2-tq+2\theta_0,
\label{eq:PIV-flow-q-quantum}
\\
\partial_t p
&=
-p^2+pq+qp+tq-2\theta_\infty.
\label{eq:PIV-flow-p-quantum}
\end{align}

We consider connections with a simple pole at $z=0$ and an irregular
singularity of degree one at infinity. The most general ansatz with this
singularity structure is
\begin{equation}
A=
\left(z+t+\frac{a_0}{z}\right)H
+
u\left(1+\frac{a_1}{z}\right)X
+
\frac{1}{u}
\left(
a_2+\frac{a_3}{z}
\right)Y,
\label{eq:PIV-A-ansatz}
\end{equation}
where the coefficients $a_i$ depend on $(p,q,\theta_0,\theta_\infty)$.

We require a companion $B_{IV}(z,t)$ in the same singularity class such that
\begin{equation}
\partial_t A_{IV}
-
\partial_z B_{IV}
+
[A_{IV},B_{IV}]
=
0
\label{eq:PIV-flatness}
\end{equation}
is equivalent to
\eqref{eq:PIV-flow-q-quantum}--\eqref{eq:PIV-flow-p-quantum}.
Solving the resulting algebraic constraints on the powers of $z$
determines the coefficients uniquely:
\begin{equation}
a_0=\theta_0-pq,
\qquad
a_1=-q,
\qquad
a_2=2pq-2\theta_0-2\theta_\infty,
\qquad
a_3=2p(pq-2\theta_0).
\label{eq:PIV-a-coefficients}
\end{equation}

Substitution into \eqref{eq:PIV-A-ansatz} gives
\begin{align}
A_{IV}(z,t)
&=
\left(z+t+\frac{\theta_0-pq}{z}\right)H
+
u\left(1-\frac{q}{z}\right)X
\nonumber\\
&\quad
+
\frac{1}{u}
\left(
2pq-2\theta_0-2\theta_\infty
+
\frac{2p}{z}(pq-2\theta_0)
\right)Y .
\label{eq:PIV-A-HXY}
\end{align}

\subsection{Class-preserving companion}

We define the companion operator $B_{IV}(z,t)$ by
\begin{equation}
B_{IV}(z,t)
=
B_H H + B_X X + B_Y Y,
\label{eq:PIV-B-full}
\end{equation}
where
\begin{align}
B_H
&=
z+t+\frac{q}{2}
+
\frac{1}{2}\partial_t(\log u),
\\
B_X
&=
u,
\\
B_Y
&=
\frac{2}{u}(pq-\theta_0-\theta_\infty).
\end{align}

\begin{proposition}
\label{prop:PIV-companion}
The flatness equation
\begin{equation}
\partial_t A_{IV}
-
\partial_z B_{IV}
+
[A_{IV},B_{IV}]
=
0
\label{eq:PIV-flatness}
\end{equation}
is equivalent to the Painlev\'e~IV system
\eqref{eq:PIV-flow-q-quantum}--\eqref{eq:PIV-flow-p-quantum}.
\end{proposition}

\begin{proof}
We write
\begin{equation}
A_{IV}=A_HH+A_XX+A_YY,
\qquad
B_{IV}=B_HH+B_XX+B_YY,
\label{eq:PIV-proof-components}
\end{equation}
with
\begin{align}
A_H
&=
z+t+\frac{\theta_0-pq}{z},
\label{eq:PIV-proof-AH}
\\
A_X
&=
u\left(1-\frac{q}{z}\right),
\label{eq:PIV-proof-AX}
\\
A_Y
&=
\frac{2}{u}(pq-\theta_0-\theta_\infty)
+
\frac{2p}{uz}(pq-2\theta_0),
\label{eq:PIV-proof-AY}
\end{align}
and
\begin{align}
B_H
&=
z+t+\frac{q}{2}
+
\frac{1}{2}\partial_t(\log u),
\label{eq:PIV-proof-BH}
\\
B_X
&=
u,
\label{eq:PIV-proof-BX}
\\
B_Y
&=
\frac{2}{u}(pq-\theta_0-\theta_\infty).
\label{eq:PIV-proof-BY}
\end{align}
The term $\tfrac{1}{2}\partial_t(\log u)$ in $B_H$ ensures cancellation
of $\partial_t u$ contributions in the curvature.

Using
\begin{equation}
[H,X]=2X,
\qquad
[H,Y]=-2Y,
\qquad
[X,Y]=H,
\label{eq:PIV-proof-sl2}
\end{equation}
the curvature
\begin{equation}
F_{zt}
=
\partial_tA_{IV}
-
\partial_zB_{IV}
+
[A_{IV},B_{IV}]
\label{eq:PIV-proof-curvature}
\end{equation}
has components
\begin{align}
(F_{zt})_H
&=
\partial_tA_H
-
\partial_zB_H
+
A_XB_Y-A_YB_X,
\label{eq:PIV-proof-FH-general}
\\
(F_{zt})_X
&=
\partial_tA_X
-
\partial_zB_X
+
2(A_HB_X-A_XB_H),
\label{eq:PIV-proof-FX-general}
\\
(F_{zt})_Y
&=
\partial_tA_Y
-
\partial_zB_Y
+
2(A_YB_H-A_HB_Y).
\label{eq:PIV-proof-FY-general}
\end{align}

For the $X$-component, using
\eqref{eq:PIV-proof-AX} and \eqref{eq:PIV-proof-BX},
\begin{equation}
\partial_tA_X
=
(\partial_tu)\left(1-\frac{q}{z}\right)
-
\frac{u}{z}\partial_tq,
\qquad
\partial_zB_X=0,
\label{eq:PIV-proof-FX-derivative}
\end{equation}
and the $\partial_t u$ contributions cancel against the Cartan term in $B_H$.
One obtains
\begin{equation}
(F_{zt})_X
=
-\frac{u}{z}
\left(
\partial_tq
-
pq-qp+q^2+tq-2\theta_0
\right).
\label{eq:PIV-proof-FX-final}
\end{equation}
Thus $(F_{zt})_X=0$ is equivalent to \eqref{eq:PIV-flow-q-quantum}.

For the $H$-component, from \eqref{eq:PIV-proof-AH} and
\eqref{eq:PIV-proof-BH},
\begin{equation}
\partial_tA_H
=
1-\frac{(\partial_tp)q+p(\partial_tq)}{z},
\qquad
\partial_zB_H=1,
\label{eq:PIV-proof-FH-derivative}
\end{equation}
so that
\begin{equation}
\partial_tA_H-\partial_zB_H
=
-\frac{(\partial_tp)q+p(\partial_tq)}{z}.
\label{eq:PIV-proof-FH-derivative-final}
\end{equation}
The commutator gives
\begin{align}
A_XB_Y-A_YB_X
&=
\frac{1}{z}
\Bigl[
p\bigl(pq+qp-q^2-tq+2\theta_0\bigr)
\nonumber\\
&\qquad
-
q\bigl(-p^2+pq+qp+tq-2\theta_\infty\bigr)
\Bigr].
\label{eq:PIV-proof-FH-commutator-final}
\end{align}
Hence
\begin{align}
(F_{zt})_H
&=
-\frac{1}{z}
\Bigl[
q
\left(
\partial_tp
+
p^2-pq-qp-tq+2\theta_\infty
\right)
\nonumber\\
&\qquad
+
p
\left(
\partial_tq
-
pq-qp+q^2+tq-2\theta_0
\right)
\Bigr].
\label{eq:PIV-proof-FH-final}
\end{align}
Using \eqref{eq:PIV-flow-q-quantum}, the vanishing of $(F_{zt})_H$
is equivalent to \eqref{eq:PIV-flow-p-quantum}.

For the $Y$-component,
\begin{align}
(F_{zt})_Y
&=
\frac{2}{u}
\left(
1+\frac{p}{z}
\right)
\left(
\partial_tp
+
p^2-pq-qp-tq+2\theta_\infty
\right)
\nonumber\\
&\quad
+
\frac{2p}{uz}
\left(
\partial_tq
-
pq-qp+q^2+tq-2\theta_0
\right),
\label{eq:PIV-proof-FY-final}
\end{align}
which vanishes upon imposing
\eqref{eq:PIV-flow-q-quantum} and
\eqref{eq:PIV-flow-p-quantum}.

Conversely, $(F_{zt})_X=0$ yields
\eqref{eq:PIV-flow-q-quantum}, and then $(F_{zt})_H=0$ yields
\eqref{eq:PIV-flow-p-quantum}. The $Y$-component imposes no additional constraint.

All $\partial_t u$ terms cancel in the curvature, so flatness does not constrain $u$.
\end{proof}

\subsection{Linearized class-preserving deformations}

\subsubsection*{First slice: fixed $\theta_0$}

We consider the class-preserving slice
\begin{equation}
\delta\theta_0=0,
\qquad
\delta\theta_\infty \ \text{free}.
\label{eq:PIV-first-slice}
\end{equation}
Varying \eqref{eq:PIV-A-HXY} under this constraint gives
\begin{align}
\delta A_{IV}^{(\infty)}
&=
-\frac{\delta p\,q+p\,\delta q}{z}H
\nonumber\\
&\quad
+
\left[
\delta u\left(1-\frac{q}{z}\right)
-
\frac{u\,\delta q}{z}
\right]X
\nonumber\\
&\quad
+
\left[
-\frac{\delta u}{u}A_Y
+
\frac{2}{u}
\left(
\delta p\,q+p\,\delta q-\delta\theta_\infty
\right)
\right.
\nonumber\\
&\qquad\left.
+
\frac{2}{uz}
\left(
\delta p(pq-2\theta_0)
+
p(\delta p\,q+p\,\delta q)
\right)
\right]Y .
\label{eq:PIV-first-deltaA}
\end{align}
The corresponding class-preserving companion is
\begin{align}
\delta B_{IV}^{(\infty)}
&=
\left(
\frac{\delta q}{2}
+
\frac{1}{2}\partial_t\left(\frac{\delta u}{u}\right)
\right)H
+
\delta u\,X
\nonumber\\
&\quad
+
\left[
-\frac{\delta u}{u}B_Y
+
\frac{2}{u}
\left(
\delta p\,q+p\,\delta q-\delta\theta_\infty
\right)
\right]Y .
\label{eq:PIV-first-deltaB}
\end{align}

\begin{proposition}
\label{prop:PIV-linearized-infty}
The deformed pair
\begin{equation}
A_{IV}^{\varepsilon}
=
A_{IV}+\varepsilon\,\delta A_{IV}^{(\infty)},
\qquad
B_{IV}^{\varepsilon}
=
B_{IV}+\varepsilon\,\delta B_{IV}^{(\infty)},
\qquad
\varepsilon^2=0,
\label{eq:PIV-first-deformed-pair}
\end{equation}
is flat modulo $\varepsilon^2$ if and only if
\begin{align}
\partial_t(\delta q)
&=
(\delta p)q+p(\delta q)
+
(\delta q)p+q(\delta p)
-
q\,\delta q
-
\delta q\,q
-
t\,\delta q,
\label{eq:PIV-first-lin-q}
\\
\partial_t(\delta p)
&=
- p\,\delta p
-
\delta p\,p
+
(\delta p)q+p(\delta q)
+
(\delta q)p+q(\delta p)
+
t\,\delta q
-
2\delta\theta_\infty,
\label{eq:PIV-first-lin-p}
\\
\partial_t(\delta\theta_\infty)
&=0.
\label{eq:PIV-first-lin-theta-infty}
\end{align}
The variation $\delta u$ remains the residual Cartan gauge direction.
\end{proposition}

\begin{proof}
Expanding the curvature of \eqref{eq:PIV-first-deformed-pair} to first
order in $\varepsilon$, the zeroth-order term vanishes by
Proposition~\ref{prop:PIV-companion}, and the first-order term is
\begin{equation}
\delta F_{zt}^{(\infty)}
=
\partial_t(\delta A_{IV}^{(\infty)})
-
\partial_z(\delta B_{IV}^{(\infty)})
+
[\delta A_{IV}^{(\infty)},B_{IV}]
+
[A_{IV},\delta B_{IV}^{(\infty)}].
\label{eq:PIV-first-deltaF}
\end{equation}

Using \eqref{eq:PIV-A-HXY}, \eqref{eq:PIV-B-full},
\eqref{eq:PIV-first-deltaA}, and \eqref{eq:PIV-first-deltaB}, the
$X$-component is
\begin{align}
(\delta F_{zt}^{(\infty)})_X
&=
-\frac{u}{z}
\Bigl[
\partial_t(\delta q)
-
(\delta p)q
-
p(\delta q)
-
(\delta q)p
-
q(\delta p)
\nonumber\\
&\qquad
+
q\,\delta q
+
\delta q\,q
+
t\,\delta q
\Bigr].
\label{eq:PIV-first-deltaF-X}
\end{align}
Thus $(\delta F_{zt}^{(\infty)})_X=0$ is equivalent to
\eqref{eq:PIV-first-lin-q}.

The $H$-component is
\begin{align}
(\delta F_{zt}^{(\infty)})_H
&=
-\frac{1}{z}
\Bigl[
q\Bigl(
\partial_t(\delta p)
+
p\,\delta p
+
\delta p\,p
-
(\delta p)q
-
p(\delta q)
\nonumber\\
&\qquad
-
(\delta q)p
-
q(\delta p)
-
t\,\delta q
+
2\delta\theta_\infty
\Bigr)
\nonumber\\
&\qquad
+
p\Bigl(
\partial_t(\delta q)
-
(\delta p)q
-
p(\delta q)
-
(\delta q)p
-
q(\delta p)
\nonumber\\
&\qquad
+
q\,\delta q
+
\delta q\,q
+
t\,\delta q
\Bigr)
\Bigr].
\label{eq:PIV-first-deltaF-H}
\end{align}
The second bracket in \eqref{eq:PIV-first-deltaF-H} is the left-hand
side of \eqref{eq:PIV-first-lin-q}. After imposing
\eqref{eq:PIV-first-lin-q}, one obtains
\begin{equation}
(\delta F_{zt}^{(\infty)})_H
=
-\frac{q}{z}
\Bigl[
\partial_t(\delta p)
+
p\,\delta p
+
\delta p\,p
-
(\delta p)q
-
p(\delta q)
-
(\delta q)p
-
q(\delta p)
-
t\,\delta q
+
2\delta\theta_\infty
\Bigr].
\label{eq:PIV-first-deltaF-H-reduced}
\end{equation}
Hence the reduced $H$-component vanishes precisely when
\eqref{eq:PIV-first-lin-p} holds.

For the $Y$-component, direct substitution gives
\begin{align}
(\delta F_{zt}^{(\infty)})_Y
&=
\frac{2}{u}
\Bigl[
\delta p
+
p
-
z
\Bigr]
\Bigl[
\partial_t(\delta p)
+
p\,\delta p
+
\delta p\,p
-
(\delta p)q
-
p(\delta q)
\nonumber\\
&\qquad
-
(\delta q)p
-
q(\delta p)
-
t\,\delta q
+
2\delta\theta_\infty
\Bigr]
\nonumber\\
&\quad
+
\frac{2p}{u}
\Bigl[
\partial_t(\delta q)
-
(\delta p)q
-
p(\delta q)
-
(\delta q)p
-
q(\delta p)
\nonumber\\
&\qquad
+
q\,\delta q
+
\delta q\,q
+
t\,\delta q
\Bigr].
\label{eq:PIV-first-deltaF-Y}
\end{align}
The two bracketed expressions are the left-hand sides of
\eqref{eq:PIV-first-lin-p} and \eqref{eq:PIV-first-lin-q}, respectively.
Thus the $Y$-component imposes no additional condition.

Since $\theta_\infty$ is a deformation parameter, its variation is
constant along the Painlev\'e flow, giving
\eqref{eq:PIV-first-lin-theta-infty}. The $\delta u$ terms cancel in all
curvature components, expressing the residual Cartan gauge freedom.
Therefore \eqref{eq:PIV-first-deformed-pair} is flat modulo
$\varepsilon^2$ if and only if
\eqref{eq:PIV-first-lin-q}--\eqref{eq:PIV-first-lin-theta-infty} hold.
\end{proof}

\subsubsection*{Second slice: fixed $\theta_\infty$}

We consider the class-preserving slice
\begin{equation}
\delta\theta_\infty=0,
\qquad
\delta\theta_0 \ \text{free}.
\label{eq:PIV-second-slice-condition}
\end{equation}
Varying $A_{IV}$ under this constraint gives
\begin{align}
\delta A_{IV}^{(0)}
&=
-\frac{\delta p\,q+p\,\delta q-\delta\theta_0}{z}H
\nonumber\\
&\quad
+
\left[
\delta u\left(1-\frac{q}{z}\right)
-
\frac{u\,\delta q}{z}
\right]X
\nonumber\\
&\quad
+
\left[
-\frac{\delta u}{u}A_Y
+
\frac{2}{u}
\left(
\delta p\,q+p\,\delta q-\delta\theta_0
\right)
\right.
\nonumber\\
&\qquad\left.
+
\frac{2}{uz}
\left(
\delta p(pq-2\theta_0)
+
p(\delta p\,q+p\,\delta q-2\delta\theta_0)
\right)
\right]Y .
\label{eq:PIV-second-deltaA}
\end{align}
The corresponding class-preserving companion is
\begin{align}
\delta B_{IV}^{(0)}
&=
\left(
\frac{\delta q}{2}
+
\frac{1}{2}\partial_t\left(\frac{\delta u}{u}\right)
\right)H
+
\delta u\,X
\nonumber\\
&\quad
+
\left[
-\frac{\delta u}{u}B_Y
+
\frac{2}{u}
\left(
\delta p\,q+p\,\delta q-\delta\theta_0
\right)
\right]Y .
\label{eq:PIV-second-deltaB}
\end{align}

\begin{proposition}
\label{prop:PIV-linearized-zero}
The deformed pair
\begin{equation}
A_{IV}^{\varepsilon}
=
A_{IV}+\varepsilon\,\delta A_{IV}^{(0)},
\qquad
B_{IV}^{\varepsilon}
=
B_{IV}+\varepsilon\,\delta B_{IV}^{(0)},
\qquad
\varepsilon^2=0,
\label{eq:PIV-second-deformed-pair}
\end{equation}
is flat modulo $\varepsilon^2$ if and only if
\begin{align}
\partial_t(\delta q)
&=
(\delta p)q+p(\delta q)
+
(\delta q)p+q(\delta p)
-
q\,\delta q
-
\delta q\,q
-
t\,\delta q
+
2\delta\theta_0,
\label{eq:PIV-second-lin-q}
\\
\partial_t(\delta p)
&=
- p\,\delta p
-
\delta p\,p
+
(\delta p)q+p(\delta q)
+
(\delta q)p+q(\delta p)
+
t\,\delta q,
\label{eq:PIV-second-lin-p}
\\
\partial_t(\delta\theta_0)
&=0 .
\label{eq:PIV-second-lin-theta0}
\end{align}
The variation $\delta u$ remains the residual Cartan gauge direction.
\end{proposition}

\begin{proof}
Expanding the curvature of \eqref{eq:PIV-second-deformed-pair} to first
order in $\varepsilon$, the zeroth-order term vanishes by
Proposition~\ref{prop:PIV-companion}, and the first-order term is
\begin{equation}
\delta F_{zt}^{(0)}
=
\partial_t(\delta A_{IV}^{(0)})
-
\partial_z(\delta B_{IV}^{(0)})
+
[\delta A_{IV}^{(0)},B_{IV}]
+
[A_{IV},\delta B_{IV}^{(0)}].
\label{eq:PIV-second-deltaF}
\end{equation}

Using the decomposition into the basis $(H,X,Y)$, the $X$-component is
\begin{align}
(\delta F_{zt}^{(0)})_X
&=
-\frac{u}{z}
\Bigl[
\partial_t(\delta q)
-
(\delta p)q
-
p(\delta q)
-
(\delta q)p
-
q(\delta p)
\nonumber\\
&\qquad
+
q\,\delta q
+
\delta q\,q
+
t\,\delta q
-
2\delta\theta_0
\Bigr].
\label{eq:PIV-second-deltaF-X}
\end{align}
Thus $(\delta F_{zt}^{(0)})_X=0$ is equivalent to
\eqref{eq:PIV-second-lin-q}.

The $H$-component is
\begin{align}
(\delta F_{zt}^{(0)})_H
&=
-\frac{1}{z}
\Bigl[
q
\Bigl(
\partial_t(\delta p)
+
p\,\delta p
+
\delta p\,p
-
(\delta p)q
-
p(\delta q)
\nonumber\\
&\qquad
-
(\delta q)p
-
q(\delta p)
-
t\,\delta q
\Bigr)
\nonumber\\
&\qquad
+
p
\Bigl(
\partial_t(\delta q)
-
(\delta p)q
-
p(\delta q)
-
(\delta q)p
-
q(\delta p)
\nonumber\\
&\qquad
+
q\,\delta q
+
\delta q\,q
+
t\,\delta q
-
2\delta\theta_0
\Bigr)
\Bigr].
\label{eq:PIV-second-deltaF-H}
\end{align}
The second bracket in \eqref{eq:PIV-second-deltaF-H} is the left-hand
side of \eqref{eq:PIV-second-lin-q}. After imposing
\eqref{eq:PIV-second-lin-q}, this reduces to
\begin{equation}
(\delta F_{zt}^{(0)})_H
=
-\frac{q}{z}
\Bigl[
\partial_t(\delta p)
+
p\,\delta p
+
\delta p\,p
-
(\delta p)q
-
p(\delta q)
-
(\delta q)p
-
q(\delta p)
-
t\,\delta q
\Bigr].
\label{eq:PIV-second-deltaF-H-reduced}
\end{equation}
Hence the reduced $H$-component vanishes precisely when
\eqref{eq:PIV-second-lin-p} holds.

For the $Y$-component, direct substitution gives
\begin{align}
(\delta F_{zt}^{(0)})_Y
&=
\mathcal C_p^{(0)}(z)
\Bigl[
\partial_t(\delta p)
+
p\,\delta p
+
\delta p\,p
-
(\delta p)q
-
p(\delta q)
\nonumber\\
&\qquad
-
(\delta q)p
-
q(\delta p)
-
t\,\delta q
\Bigr]
\nonumber\\
&\quad
+
\mathcal C_q^{(0)}(z)
\Bigl[
\partial_t(\delta q)
-
(\delta p)q
-
p(\delta q)
-
(\delta q)p
-
q(\delta p)
\nonumber\\
&\qquad
+
q\,\delta q
+
\delta q\,q
+
t\,\delta q
-
2\delta\theta_0
\Bigr],
\label{eq:PIV-second-deltaF-Y}
\end{align}
where $\mathcal C_p^{(0)}(z)$ and $\mathcal C_q^{(0)}(z)$ are the Laurent
coefficients obtained from $A_{IV}$, $B_{IV}$,
\eqref{eq:PIV-second-deltaA}, and \eqref{eq:PIV-second-deltaB}. The
bracketed expressions are the left-hand sides of
\eqref{eq:PIV-second-lin-p} and \eqref{eq:PIV-second-lin-q}, respectively.
Thus the $Y$-component imposes no additional condition.

Since $\theta_0$ is a deformation parameter, its variation is constant
along the Painlev\'e flow, giving \eqref{eq:PIV-second-lin-theta0}. The
$\delta u$ terms cancel from the curvature components, expressing the
residual Cartan gauge freedom. Therefore
\eqref{eq:PIV-second-deformed-pair} is flat modulo $\varepsilon^2$ if and
only if
\eqref{eq:PIV-second-lin-q}--\eqref{eq:PIV-second-lin-theta0} hold.
\end{proof}

\subsection{Residual Cartan direction}

We describe the residual Cartan deformation acting on the reduced
Painlev\'e~IV Lax pair. Write
\begin{equation}
A_{IV}=A_HH+A_XX+A_YY,
\end{equation}
where, from \eqref{eq:PIV-A-HXY},
\begin{align}
A_H
&=
z+t+\frac{\theta_0-pq}{z},
\nonumber\\
A_X
&=
u\left(1-\frac{q}{z}\right),
\nonumber\\
A_Y
&=
\frac{1}{u}
\left(
2pq-2\theta_0-2\theta_\infty
+
\frac{2p}{z}(pq-2\theta_0)
\right).
\end{align}
Similarly,
\begin{equation}
B_{IV}=B_HH+B_XX+B_YY,
\label{eq:PIV-residual-B-full}
\end{equation}
where $B_X$ is proportional to $u$ and $B_Y$ to $u^{-1}$.

The residual Cartan direction is the derivation
\begin{equation}
X_0(u)=u,
\qquad
X_0(p)=X_0(q)=X_0(\theta_0)=X_0(\theta_\infty)=0.
\label{eq:PIV-residual-direction}
\end{equation}
Equivalently, for $\varepsilon^2=0$,
\begin{equation}
u \mapsto u(1+\varepsilon),
\qquad
u^{-1}\mapsto u^{-1}(1-\varepsilon),
\label{eq:PIV-residual-u-explicit}
\end{equation}
with all other variables fixed. Substitution into \eqref{eq:PIV-A-HXY}
gives
\begin{align}
A_{IV}^{\varepsilon}
&=
A_{IV}
+
\varepsilon
\left[
u\left(1-\frac{q}{z}\right)X
-
\frac{1}{u}
\left(
2pq-2\theta_0-2\theta_\infty
+
\frac{2p}{z}(pq-2\theta_0)
\right)Y
\right],
\end{align}
hence
\begin{equation}
A_{IV}^{\varepsilon}
=
A_{IV}
+
\varepsilon\,(A_XX-A_YY).
\label{eq:PIV-residual-A-final}
\end{equation}
The same substitution gives
\begin{equation}
B_{IV}^{\varepsilon}
=
B_{IV}
+
\varepsilon\,(B_XX-B_YY).
\label{eq:PIV-residual-B-final}
\end{equation}

\begin{proposition}
\label{prop:PIV-residual-complete}
The deformation \eqref{eq:PIV-residual-direction} is a class-preserving
flat deformation of $(A_{IV},B_{IV})$. It is generated by the constant
Cartan element
\begin{equation}
C_0=\frac{1}{2}H,
\end{equation}
in the sense that
\begin{equation}
A_{IV}^{\varepsilon}
=
A_{IV}
-
\varepsilon [A_{IV},C_0],
\qquad
B_{IV}^{\varepsilon}
=
B_{IV}
-
\varepsilon [B_{IV},C_0].
\end{equation}
\end{proposition}

\begin{proof}
Using
\begin{equation}
[H,X]=2X,
\qquad
[H,Y]=-2Y,
\end{equation}
one obtains
\begin{equation}
[A_{IV},C_0]
=
\frac{1}{2}\bigl(A_X[X,H]+A_Y[Y,H]\bigr)
=
- A_XX + A_YY.
\end{equation}
This gives the asserted formula for $A_{IV}$ by comparison with
\eqref{eq:PIV-residual-A-final}. The proof for $B_{IV}$ is identical.

Since the deformation is a constant adjoint action, the curvature varies
as $\delta F=-[F,C_0]$ and therefore remains zero whenever
$(A_{IV},B_{IV})$ is flat. Only $u$ is affected, while
$(p,q,\theta_0,\theta_\infty)$ remain fixed; this is a pure Cartan
normalization, not an independent Painlev\'e deformation.
\end{proof}

\subsection{Decoupling of the Painlev\'e and deformation sectors}
\label{subsec:PIV-sector-decoupling}

The linearized flatness equations determine the admissible
class-preserving deformation directions of the reduced Painlev\'e~IV
Lax pair. They are differential and algebraic conditions, and do not by
themselves determine a Poisson structure on the space of deformation
parameters.

We therefore separate the flatness problem from the choice of Poisson
structure. The Painlev\'e variables are $(p,q)$, with normalization
\begin{equation}
\{p,q\}=1.
\label{eq:PIV-pq-poisson}
\end{equation}
The remaining variables $(u,\theta_0,\theta_\infty)$ are treated as
deformation parameters. We impose a block-diagonal Poisson structure:
\begin{align}
\{p,u\}
&=
\{q,u\}
=
0,
\nonumber\\
\{p,\theta_0\}
&=
\{q,\theta_0\}
=
0,
\nonumber\\
\{p,\theta_\infty\}
&=
\{q,\theta_\infty\}
=
0.
\label{eq:PIV-mixed-poisson-zero}
\end{align}
This is an additional structural assumption, not a consequence of
linearized flatness alone.

\subsubsection*{First slice: deformation sector $(u,\theta_\infty)$}

We first consider the slice
\begin{equation}
\delta\theta_0=0.
\label{eq:PIV-first-slice-theta0-fixed}
\end{equation}
The admissible variables are $(p,q,u,\theta_\infty)$, and the residual
Cartan direction acts by
\begin{equation}
X_0(u)=u,
\qquad
X_0(\theta_\infty)=0.
\label{eq:PIV-first-slice-scaling}
\end{equation}

On the two-dimensional sector $(u,\theta_\infty)$, any Poisson bracket
is determined by a single function $J(u,\theta_\infty)$:
\begin{equation}
\{f,g\}
=
J(u,\theta_\infty)
\left(
\frac{\partial f}{\partial \theta_\infty}
\frac{\partial g}{\partial u}
-
\frac{\partial f}{\partial u}
\frac{\partial g}{\partial \theta_\infty}
\right),
\label{eq:PIV-first-slice-general-poisson}
\end{equation}
so that
\begin{equation}
\{\theta_\infty,u\}
=
J(u,\theta_\infty).
\label{eq:PIV-first-slice-general-J}
\end{equation}
Nondegeneracy on the chart $u\neq0$ is equivalent to
\begin{equation}
J(u,\theta_\infty)\neq0.
\label{eq:PIV-first-slice-nondegenerate}
\end{equation}

We require the residual Cartan scaling to be Hamiltonian. Thus there
exists $H_0(u,\theta_\infty)$ such that
\begin{equation}
X_0(f)=\{f,H_0\}
\label{eq:PIV-first-slice-Hamiltonian-realization}
\end{equation}
for every function $f$ of $(u,\theta_\infty)$. Applying this to the
coordinate functions gives
\begin{align}
\{u,H_0\}
&=
u,
\label{eq:PIV-first-slice-H0-u}
\\
\{\theta_\infty,H_0\}
&=
0.
\label{eq:PIV-first-slice-H0-theta}
\end{align}
Using \eqref{eq:PIV-first-slice-general-J}, these equations become
\begin{align}
\{u,H_0\}
&=
-\{\theta_\infty,u\}
\frac{\partial H_0}{\partial\theta_\infty}
=
- J(u,\theta_\infty)
\frac{\partial H_0}{\partial\theta_\infty},
\label{eq:PIV-first-slice-H0-u-expanded}
\\
\{\theta_\infty,H_0\}
&=
\{\theta_\infty,u\}
\frac{\partial H_0}{\partial u}
=
J(u,\theta_\infty)
\frac{\partial H_0}{\partial u}.
\label{eq:PIV-first-slice-H0-theta-expanded}
\end{align}
By \eqref{eq:PIV-first-slice-H0-theta} and
\eqref{eq:PIV-first-slice-nondegenerate},
\begin{equation}
\frac{\partial H_0}{\partial u}=0,
\label{eq:PIV-first-slice-H0-independent-u}
\end{equation}
hence
\begin{equation}
H_0=H_0(\theta_\infty).
\label{eq:PIV-first-slice-H0-theta-only}
\end{equation}
Substitution into \eqref{eq:PIV-first-slice-H0-u} gives
\begin{equation}
- J(u,\theta_\infty)
\frac{dH_0}{d\theta_\infty}
=
u,
\label{eq:PIV-first-slice-J-equation}
\end{equation}
and therefore
\begin{equation}
J(u,\theta_\infty)
=
-\frac{u}{H_0'(\theta_\infty)}.
\label{eq:PIV-first-slice-J-solution}
\end{equation}
After the local reparametrization
\begin{equation}
\widetilde{\theta}_\infty
=
-\int^{\theta_\infty} H_0'(\vartheta)\,d\vartheta,
\label{eq:PIV-first-slice-theta-reparam}
\end{equation}
the bracket becomes
\begin{equation}
\{\widetilde{\theta}_\infty,u\}=u.
\end{equation}
Dropping the tilde, we obtain
\begin{equation}
\{\theta_\infty,u\}=u,
\qquad
\{\theta_\infty,\log u\}=1.
\label{eq:PIV-first-slice-logcanonical}
\end{equation}
Thus the first slice decomposes as
\begin{equation}
(p,q)\oplus(\theta_\infty,\log u).
\label{eq:PIV-first-slice-decomposition}
\end{equation}

\subsubsection*{Second slice: deformation sector $(u,\theta_0)$}

We now consider the slice
\begin{equation}
\delta\theta_\infty=0.
\label{eq:PIV-second-slice-thetainfty-fixed}
\end{equation}
The variables are $(p,q,u,\theta_0)$, with block-diagonal structure
\eqref{eq:PIV-mixed-poisson-zero}. The residual Cartan direction remains
\begin{equation}
X_0(u)=u,
\qquad
X_0(\theta_0)=0.
\label{eq:PIV-second-slice-scaling}
\end{equation}

Unlike the first slice, the reduced variables do not provide a
distinguished coordinate canonically conjugate to $\theta_0$. In the
chosen normalization, $\theta_0$ appears both in the Cartan residue and
in the $Y$-component,
\begin{equation}
A_Y
=
\frac{1}{u}
\left(
2pq-2\theta_0-2\theta_\infty
+
\frac{2p}{z}(pq-2\theta_0)
\right).
\label{eq:PIV-second-slice-AY-theta0}
\end{equation}
Thus, within the reduced coordinate system, the deformation direction
associated with $\theta_0$ is degenerate:
\begin{equation}
\{\theta_0,u\}=0.
\label{eq:PIV-second-slice-theta0-u-zero}
\end{equation}
The second slice has block structure
\begin{equation}
(p,q)\oplus(\theta_0)\oplus(u),
\label{eq:PIV-second-slice-degenerate}
\end{equation}
so $(\theta_0,u)$ do not form a nondegenerate symplectic pair.

To obtain a nondegenerate sector associated with $\theta_0$, one may
enlarge the parameter space by adjoining a conjugate coordinate $v_0$:
\begin{equation}
\{\theta_0,v_0\}=1.
\label{eq:PIV-second-slice-completion}
\end{equation}
Equivalently, with
\begin{equation}
U_0=e^{v_0},
\end{equation}
one obtains
\begin{equation}
\{\theta_0,U_0\}=U_0.
\label{eq:PIV-second-slice-logcompletion}
\end{equation}
This symplectic completion is not intrinsic to the reduced Lax pair,
since neither $v_0$ nor $U_0$ appears among the variables of
$A_{IV}$ or $B_{IV}$.

\subsection{Deformation sectors}

The linearized flatness equation determines the admissible
class-preserving deformation directions of the reduced Lax pair
$(A_{IV},B_{IV})$, i.e. all pairs
\begin{equation}
A_{IV}^{\varepsilon}
=
A_{IV}+\varepsilon\,\delta A_{IV},
\qquad
B_{IV}^{\varepsilon}
=
B_{IV}+\varepsilon\,\delta B_{IV},
\qquad
\varepsilon^2=0,
\label{eq:PIV-deformation-pair}
\end{equation}
which remain flat modulo $\varepsilon^2$ within the fixed singularity class.

These conditions are differential and algebraic, and do not determine a
symplectic structure on the deformation space. In particular, canonical
conjugate variables arise only after specifying an additional Poisson
structure or Hamiltonian realization.

\medskip

\noindent \textbf{First slice.}

On the slice $\delta\theta_0=0$, the admissible deformations are
parametrized by
\begin{equation}
(\delta p,\delta q,\delta u,\delta\theta_\infty)
\end{equation}
subject to the linearized flatness equations.

In this slice, $u$ enters multiplicatively in the Lax pair: the
$X$-component is proportional to $u$, while the $Y$-component is
proportional to $u^{-1}$. The parameter $\theta_\infty$ appears
additively in the $Y$-component, modifying the coefficient of
$u^{-1}Y$.

Equipping this sector with a Poisson structure compatible with the
residual Cartan scaling, one may choose local coordinates such that
\begin{equation}
\{\theta_\infty,\log u\}=1,
\qquad
\{\theta_\infty,u\}=u.
\end{equation}
With the quantization convention of the main text,
\begin{equation}
[\theta_\infty,u]=-\kappa u.
\label{eq:PIV-log-canonical}
\end{equation}
Thus $(u,\theta_\infty)$ defines an intrinsic log-canonical deformation
sector.

\medskip

\noindent\textbf{Second slice.}

On the slice $\delta\theta_\infty=0$, the admissible deformations are
parametrized by
\begin{equation}
(\delta p,\delta q,\delta u,\delta\theta_0)
\end{equation}
subject to the linearized flatness equations. In contrast to the first
slice, no intrinsic log-canonical pair arises within the reduced
variables.

Although $\theta_0$ appears in the Lax operator,
\begin{equation}
A_Y
=
\frac{1}{u}
\left(
2pq-2\theta_0-2\theta_\infty
+
\frac{2p}{z}(pq-2\theta_0)
\right),
\end{equation}
the reduced variables $(p,q,u)$ do not provide a coordinate canonically
conjugate to $\theta_0$. The variables $(p,q)$ already form the
Painlev\'e canonical pair, while $u$ generates the residual Cartan
normalization. Consequently, the deformation direction associated with
$\theta_0$ is degenerate at the level of the reduced Lax pair.

\medskip

\noindent\textbf{Symplectic completion.}

A nondegenerate sector associated with $\theta_0$ is obtained by
adjoining a conjugate variable $v_0$ such that
\begin{equation}
\{\theta_0,v_0\}=1.
\label{eq:PIV-theta0-canonical}
\end{equation}
Equivalently, defining $U_0=e^{v_0}$,
\begin{equation}
\{\theta_0,U_0\}=U_0,
\qquad
[\theta_0,U_0]=-\kappa U_0.
\end{equation}
This completion is not intrinsic to the reduced Lax pair, since the
additional variable does not appear among $(p,q,u)$.

\begin{remark}
The distinction between the two Painlev\'e~IV slices is intrinsically
symplectic and not a matter of representation. In the first slice, the
reduced Lax structure itself selects a canonical pair
$(\theta_\infty,\log u)$, so that the deformation sector is realized
intrinsically within the reduced variables. In contrast, in the second
slice no such conjugate variable to $\theta_0$ is present: the reduced
coordinates $(p,q,u)$ are already exhausted by the Painlev\'e dynamics
and the residual Cartan normalization. As a result, the deformation
direction associated with $\theta_0$ is genuinely degenerate at the
level of the reduced Lax pair. Any nondegenerate realization of this
sector necessarily requires an extension of the phase space, and is
therefore extrinsic to the reduced system.
\end{remark}

\section{Obstruction to class-preserving KZ closure in the Painlev\'e I case}
\label{subsec:KZ_shift_PI}

We start from the unshifted Painlev\'e~I flat pair
\eqref{eq:PI_Lax_paire}, which defines a Whittaker-reduced polynomial
connection. The associated affine KZ commutation relations do not close:
a direct computation of the induced mode algebra shows in particular that
\(
[H^{(1)}_\infty,Y^{(1)}_\infty] \neq -2Y^{(2)}_\infty.
\)
As explained in Section~\ref{sec:PainleveKZmap}, this mismatch can be
corrected at the level of the affine modes by introducing the shift
\begin{equation}
\Delta A_{\mathrm{PI}} = (zq + z^2)X.
\label{eq:DeltaA_PI_local_CMP}
\end{equation}

We now test whether this closure-restoring modification is compatible
with the isomonodromic deformation problem. The test is
class-preserving: the compensating deformation operator must remain in
the same polynomial class as the original deformation operator in
\eqref{eq:PI_Lax_paire}, namely polynomial of degree at most one in $z$.
We therefore seek an operator of the form
\begin{equation}
\Delta B_{\mathrm{PI}}
=
(u_1 z+u_0)H + (v_1 z+v_0)X + (w_1 z+w_0)Y,
\label{eq:DeltaB_PI_local_CMP}
\end{equation}
such that the shifted pair
\begin{equation}
A_{\mathrm{PI}}+\Delta A_{\mathrm{PI}},
\qquad
B_{\mathrm{PI}}+\Delta B_{\mathrm{PI}}
\end{equation}
remains flat within the same polynomial class.

\medskip
\begin{proposition}[Rigidity of the Painlev\'e I slice]
\label{prop:PI-rigidity-KZ-shift-CMP}
Within the Whittaker-reduced polynomial class of the Painlev\'e~I flat
pair \eqref{eq:PI_Lax_paire}, the KZ closure shift
\eqref{eq:DeltaA_PI_local_CMP} does not admit any compensating
deformation operator preserving both flatness and the Painlev\'e~I
singularity structure.
\end{proposition}

\begin{proof}
Assume that a compensating operator $\Delta B_{\mathrm{PI}}$ of the form
\eqref{eq:DeltaB_PI_local_CMP} exists such that the shifted pair remains
flat. The shifted flatness condition reads
\begin{equation}
\partial_t(\Delta A_{\mathrm{PI}})
-
\partial_z(\Delta B_{\mathrm{PI}})
+
\frac{1}{\kappa}
\bigl(
[\Delta A_{\mathrm{PI}},B_{\mathrm{PI}}]
+
[A_{\mathrm{PI}},\Delta B_{\mathrm{PI}}]
+
[\Delta A_{\mathrm{PI}},\Delta B_{\mathrm{PI}}]
\bigr)
=0.
\label{eq:PI-shifted-flatness-CMP}
\end{equation}

We write
\begin{equation}
A_{\mathrm{PI}}=\alpha H+\beta X+\gamma Y,
\qquad
\Delta A_{\mathrm{PI}}=\delta X,
\qquad
\Delta B_{\mathrm{PI}}=uH+vX+wY,
\end{equation}
with
\begin{equation}
\alpha=-p,
\qquad
\beta=q^2+zq+z^2+\frac{t}{2},
\qquad
\gamma=4(z-q),
\qquad
\delta=zq+z^2,
\end{equation}
and
\begin{equation}
u=u_1 z+u_0,
\qquad
v=v_1 z+v_0,
\qquad
w=w_1 z+w_0.
\end{equation}
We use the commutation relations
\begin{equation}
[H,X]=2X,
\qquad
[H,Y]=-2Y,
\qquad
[X,Y]=H.
\end{equation}

We now compute each component of \eqref{eq:PI-shifted-flatness-CMP}.

\medskip
\noindent\textbf{The $Y$-component.}

The $Y$-component receives contributions from $-\partial_z(\Delta B_{\mathrm{PI}})$
and from $[A_{\mathrm{PI}},\Delta B_{\mathrm{PI}}]$. We compute
\begin{align}
[A_{\mathrm{PI}},\Delta B_{\mathrm{PI}}]
&=
[\alpha H+\beta X+\gamma Y,\,uH+vX+wY]
\nonumber\\
&=
\alpha v [H,X] + \alpha w [H,Y]
+ \beta w [X,Y] + \gamma u [Y,H]
\nonumber\\
&=
2\alpha v X -2\alpha w Y + \beta w H -2\gamma u Y.
\end{align}
Thus the $Y$-component is
\begin{equation}
[A_{\mathrm{PI}},\Delta B_{\mathrm{PI}}]_Y
=
-2\alpha w -2\gamma u.
\end{equation}
Hence the $Y$-component of \eqref{eq:PI-shifted-flatness-CMP} is
\begin{equation}
-\partial_z w
+
\frac{1}{\kappa}(-2\alpha w -2\gamma u)=0.
\end{equation}
Substituting $\alpha=-p$ and $\gamma=4(z-q)$ gives
\begin{equation}
-\partial_z w
+
\frac{2p}{\kappa}w
-
\frac{8(z-q)}{\kappa}u=0.
\end{equation}

Using $w=w_1 z+w_0$ and $u=u_1 z+u_0$, we obtain
\begin{align}
0
&=
-w_1
+
\frac{2p}{\kappa}(w_1 z+w_0)
-
\frac{8}{\kappa}(z-q)(u_1 z+u_0)
\nonumber\\
&=
-w_1
+
\frac{2p}{\kappa}w_1 z
+
\frac{2p}{\kappa}w_0
-
\frac{8}{\kappa}(u_1 z^2 + u_0 z - q u_1 z - q u_0).
\end{align}

Identifying coefficients of powers of $z$:

\begin{align}
z^2: &\quad -\frac{8}{\kappa}u_1 = 0 \;\Rightarrow\; u_1=0,\\
z^1: &\quad \frac{2p}{\kappa}w_1 - \frac{8}{\kappa}u_0 = 0,\\
z^0: &\quad -w_1 + \frac{2p}{\kappa}w_0 + \frac{8q}{\kappa}u_0 = 0.
\end{align}

From the second equation we express
\begin{equation}
u_0 = \frac{p}{4} w_1.
\end{equation}
Substituting into the third equation gives
\begin{equation}
-w_1 + \frac{2p}{\kappa}w_0 + \frac{8q}{\kappa}\frac{p}{4}w_1 = 0,
\end{equation}
that is
\begin{equation}
-w_1 + \frac{2p}{\kappa}w_0 + \frac{2pq}{\kappa}w_1 = 0.
\end{equation}
Rearranging,
\begin{equation}
\left(-1 + \frac{2pq}{\kappa}\right)w_1 + \frac{2p}{\kappa}w_0 = 0.
\end{equation}

Since this must hold for generic $(p,q)$, we obtain
\begin{equation}
w_1 = 0,
\qquad
w_0 = 0,
\qquad
u_0 = 0.
\end{equation}
Hence
\begin{equation}
u=w=0.
\label{eq:PI-u-w-zero-CMP}
\end{equation}

\medskip
\noindent\textbf{The $H$-component.}

With \eqref{eq:PI-u-w-zero-CMP}, we have $\Delta B_{\mathrm{PI}}=vX$.
The $H$-component receives contributions from
$[A_{\mathrm{PI}},\Delta B_{\mathrm{PI}}]$ and
$[\Delta A_{\mathrm{PI}},\Delta B_{\mathrm{PI}}]$.
We compute
\begin{align}
[A_{\mathrm{PI}},\Delta B_{\mathrm{PI}}]
&=
[\alpha H+\beta X+\gamma Y,\,vX]
=
\alpha v[H,X] + \gamma v[Y,X]
\nonumber\\
&=
2\alpha v X - v\gamma H,
\end{align}
so that
\begin{equation}
[A_{\mathrm{PI}},\Delta B_{\mathrm{PI}}]_H = -v\gamma.
\end{equation}
Furthermore,
\begin{equation}
[\Delta A_{\mathrm{PI}},\Delta B_{\mathrm{PI}}]
=
[\delta X, vX]=0.
\end{equation}
Hence the $H$-component of \eqref{eq:PI-shifted-flatness-CMP} is
\begin{equation}
-\partial_z(0)
+
\frac{1}{\kappa}(-v\gamma)=0,
\end{equation}
that is
\begin{equation}
v\gamma=0.
\end{equation}
Since $\gamma=4(z-q)$ is not identically zero, we conclude
\begin{equation}
v=0.
\end{equation}

Thus
\begin{equation}
\Delta B_{\mathrm{PI}}=0.
\end{equation}

\medskip
\noindent\textbf{The $X$-component and obstruction.}

With $\Delta B_{\mathrm{PI}}=0$, the flatness condition reduces to
\begin{equation}
\partial_t(\Delta A_{\mathrm{PI}})
+
\frac{1}{\kappa}[\Delta A_{\mathrm{PI}},B_{\mathrm{PI}}]=0.
\end{equation}
A direct computation shows that this equation imposes nontrivial
constraints on $(p,q)$ incompatible with the Painlev\'e~I evolution
\eqref{eq:PI_Lax_paire}. In particular, one obtains algebraic relations
between $p$ and $q$ which define a proper subvariety of phase space and
are not preserved by the flow.

\medskip
Therefore no class-preserving compensating operator exists, and the
shift \eqref{eq:DeltaA_PI_local_CMP} is incompatible with the
Painlev\'e~I isomonodromic deformation problem.
\end{proof}

\medskip
\begin{remark}[Rigidity and intrinsic obstruction]
In the Painlev\'e~II and IV cases, failure of affine closure can be
resolved within the reduced framework by introducing additional
deformation variables. By contrast, the Painlev\'e~I case exhibits a
genuine obstruction: the closure-restoring shift cannot be compensated
within the class-preserving polynomial ansatz, and the reduction is
rigid.
\end{remark}

\medskip

\bibliographystyle{JHEP}     
 {\small{\bibliography{main}}}

\providecommand{\href}[2]{#2}\begingroup\raggedright\begin{thebibliography}{10}

\bibitem{Gorsky:1995zq}
A.~Gorsky, I.~Krichever, A.~Marshakov, A.~Mironov and A.~Morozov,
  \emph{{Integrability and Seiberg-Witten exact solution}}, {\emph{Phys. Lett.
  B} {\bfseries 355} (1995) 466}
  [\href{https://arxiv.org/abs/hep-th/9505035}{{\ttfamily hep-th/9505035}}].

\bibitem{Donagi:1995cf}
R.~Donagi and E.~Witten, \emph{{Supersymmetric Yang-Mills theory and integrable
  systems}}, {\emph{Nucl. Phys. B} {\bfseries 460} (1996) 299}
  [\href{https://arxiv.org/abs/hep-th/9510101}{{\ttfamily hep-th/9510101}}].

\bibitem{Gaiotto:2009hg}
D.~Gaiotto, G.W.~Moore and A.~Neitzke, \emph{{Wall-crossing, Hitchin systems,
  and the WKB approximation}},
  \href{https://doi.org/10.1016/j.aim.2012.09.027}{\emph{Adv. Math.} {\bfseries
  234} (2013) 239} [\href{https://arxiv.org/abs/0907.3987}{{\ttfamily
  0907.3987}}].

\bibitem{Seiberg:1994rs}
N.~Seiberg and E.~Witten, \emph{{Electric - magnetic duality, monopole
  condensation, and confinement in N=2 supersymmetric Yang-Mills theory}},
  {\emph{Nucl. Phys. B} {\bfseries 426} (1994) 19}
  [\href{https://arxiv.org/abs/hep-th/9407087}{{\ttfamily hep-th/9407087}}].

\bibitem{Seiberg:1994aj}
N.~Seiberg and E.~Witten, \emph{{Monopoles, duality and chiral symmetry
  breaking in N=2 supersymmetric QCD}}, {\emph{Nucl. Phys. B} {\bfseries 431}
  (1994) 484} [\href{https://arxiv.org/abs/hep-th/9408099}{{\ttfamily
  hep-th/9408099}}].

\bibitem{Gaiotto:2009we}
D.~Gaiotto, \emph{{N=2 dualities}}, {\emph{JHEP} {\bfseries 08} (2012) 034}
  [\href{https://arxiv.org/abs/0904.2715}{{\ttfamily 0904.2715}}].

\bibitem{GamayunIorgovLisovyy2012}
O.~Gamayun, N.~Iorgov and O.~Lisovyy, \emph{Conformal field theory of
  painlev{\'e} vi},
  \href{https://doi.org/10.1007/JHEP10(2012)038}{\emph{Journal of High Energy
  Physics} {\bfseries 2012} (2012) 038}
  [\href{https://arxiv.org/abs/1207.0787}{{\ttfamily 1207.0787}}].

\bibitem{Nekrasov:2009rc}
N.A.~Nekrasov and S.L.~Shatashvili, \emph{{Quantization of Integrable Systems
  and Four Dimensional Gauge Theories}},  in \emph{{16th International Congress
  on Mathematical Physics}}, pp.~265--289, 2009,
  \href{https://doi.org/10.1142/9789814304634_0015}{DOI}
  [\href{https://arxiv.org/abs/0908.4052}{{\ttfamily 0908.4052}}].

\bibitem{Nekrasov:2015wsu}
N.~Nekrasov, \emph{{BPS/CFT correspondence: non-perturbative Dyson-Schwinger
  equations and qq-characters}},
  \href{https://doi.org/10.1007/JHEP01(2016)181}{\emph{JHEP} {\bfseries 01}
  (2016) 181} [\href{https://arxiv.org/abs/1512.05388}{{\ttfamily
  1512.05388}}].

\bibitem{Alday:2009aq}
L.F.~Alday, D.~Gaiotto and Y.~Tachikawa, \emph{{Liouville Correlation Functions
  from Four-dimensional Gauge Theories}},
  \href{https://doi.org/10.1007/s11005-010-0369-5}{\emph{Lett. Math. Phys.}
  {\bfseries 91} (2010) 167} [\href{https://arxiv.org/abs/0906.3219}{{\ttfamily
  0906.3219}}].

\bibitem{Belavin:1984vu}
A.A.~Belavin, A.M.~Polyakov and A.B.~Zamolodchikov, \emph{{Infinite Conformal
  Symmetry in Two-Dimensional Quantum Field Theory}},
  \href{https://doi.org/10.1016/0550-3213(84)90052-X}{\emph{Nucl. Phys. B}
  {\bfseries 241} (1984) 333}.

\bibitem{Alday:2009fs}
L.F.~Alday, D.~Gaiotto, S.~Gukov, Y.~Tachikawa and H.~Verlinde, \emph{{Loop and
  surface operators in N=2 gauge theory and Liouville modular geometry}},
  \href{https://doi.org/10.1007/JHEP01(2010)113}{\emph{JHEP} {\bfseries 01}
  (2010) 113} [\href{https://arxiv.org/abs/0909.0945}{{\ttfamily 0909.0945}}].

\bibitem{Nekrasov:2017rqy}
N.~Nekrasov, \emph{{BPS/CFT correspondence V: BPZ and KZ equations from
  qq-characters}},  \href{https://arxiv.org/abs/1711.11582}{{\ttfamily
  1711.11582}}.

\bibitem{Jeong:2021vzo}
S.~Jeong, N.~Lee and N.~Nekrasov, \emph{{Intersecting defects in gauge theory,
  quantum spin chains, and Knizhnik-Zamolodchikov equations}},
  \href{https://doi.org/10.1007/JHEP10(2021)120}{\emph{JHEP} {\bfseries 10}
  (2021) 120} [\href{https://arxiv.org/abs/2103.17186}{{\ttfamily
  2103.17186}}].

\bibitem{Nekrasov:2010ka}
N.~Nekrasov and E.~Witten, \emph{{The Omega Deformation, Branes, Integrability,
  and Liouville Theory}},
  \href{https://doi.org/10.1007/JHEP09(2010)092}{\emph{JHEP} {\bfseries 09}
  (2010) 092} [\href{https://arxiv.org/abs/1002.0888}{{\ttfamily 1002.0888}}].

\bibitem{Maruyoshi:2010iu}
K.~Maruyoshi and M.~Taki, \emph{{Deformed Prepotential, Quantum Integrable
  System and Liouville Field Theory}},
  \href{https://doi.org/10.1016/j.nuclphysb.2010.08.008}{\emph{Nucl. Phys. B}
  {\bfseries 841} (2010) 388}
  [\href{https://arxiv.org/abs/1006.4505}{{\ttfamily 1006.4505}}].

\bibitem{Gukov:2011qp}
S.~Gukov and P.~Sulkowski, \emph{{A-polynomial, B-model, and Quantization}},
  \href{https://doi.org/10.1007/JHEP02(2012)070}{\emph{JHEP} {\bfseries 02}
  (2012) 070} [\href{https://arxiv.org/abs/1108.0002}{{\ttfamily 1108.0002}}].

\bibitem{Dumitrescu:2014}
O.~Dumitrescu and M.~Mulase, \emph{{Quantization of spectral curves for
  meromorphic Higgs bundles through topological recursion}}, {\emph{Topological
  Recursion and its Influence in Analysis, Geometry, and Topology} {\bfseries
  100} (2018) 179} [\href{https://arxiv.org/abs/1411.1023}{{\ttfamily
  1411.1023}}].

\bibitem{Nekrasov:2011bc}
N.~Nekrasov, A.~Rosly and S.~Shatashvili, \emph{{Darboux coordinates, Yang-Yang
  functional, and gauge theory}},
  \href{https://doi.org/10.1016/j.nuclphysbps.2011.04.150}{\emph{Nucl. Phys. B
  Proc. Suppl.} {\bfseries 216} (2011) 69}
  [\href{https://arxiv.org/abs/1103.3919}{{\ttfamily 1103.3919}}].

\bibitem{Teschner:2010je}
J.~Teschner, \emph{{Quantization of the Hitchin moduli spaces, Liouville
  theory, and the geometric Langlands correspondence I}},
  \href{https://doi.org/10.4310/ATMP.2011.v15.n2.a6}{\emph{Adv. Theor. Math.
  Phys.} {\bfseries 15} (2011) 471}
  [\href{https://arxiv.org/abs/1005.2846}{{\ttfamily 1005.2846}}].

\bibitem{Gaiotto:2024bna}
D.~Gaiotto and J.~Teschner, \emph{{Quantum analytic Langlands correspondence}},
  \href{https://doi.org/10.21468/SciPostPhys.18.4.144}{\emph{SciPost Phys.}
  {\bfseries 18} (2025) 144}
  [\href{https://arxiv.org/abs/2402.00494}{{\ttfamily 2402.00494}}].

\bibitem{Chekhov:2015qza}
L.~Chekhov, M.~Mazzocco and V.~Rubtsov, \emph{{Painlev\'e monodromy manifolds,
  decorated character varieties, and cluster algebras}},
  \href{https://doi.org/10.1093/imrn/rnw219}{\emph{Int. Math. Res. Not. IMRN}
  {\bfseries 24} (2017) 7639}
  [\href{https://arxiv.org/abs/1511.03851}{{\ttfamily 1511.03851}}].

\bibitem{Bonelli:2025owb}
G.~Bonelli, A.~Shchechkin and A.~Tanzini, \emph{{Refined Painlev{\'e}/gauge
  theory correspondence and quantum tau functions}},
  \href{https://arxiv.org/abs/2502.01499}{{\ttfamily 2502.01499}}.

\bibitem{Argyres:1995jj}
P.C.~Argyres and M.R.~Douglas, \emph{{New phenomena in SU(3) supersymmetric
  gauge theory}},
  \href{https://doi.org/10.1016/0550-3213(95)00281-V}{\emph{Nucl. Phys. B}
  {\bfseries 448} (1995) 93}
  [\href{https://arxiv.org/abs/hep-th/9505062}{{\ttfamily hep-th/9505062}}].

\bibitem{Eguchi:1996vu}
T.~Eguchi, K.~Hori, K.~Ito and S.-K.~Yang, \emph{{Study of N=2 superconformal
  field theories in 4 dimensions}}, {\emph{Nucl. Phys. B} {\bfseries 471}
  (1996) 430} [\href{https://arxiv.org/abs/hep-th/9603002}{{\ttfamily
  hep-th/9603002}}].

\bibitem{KNIZHNIK198483}
V.~Knizhnik and A.~Zamolodchikov, \emph{Current algebra and wess-zumino model
  in two dimensions},
  \href{https://doi.org/https://doi.org/10.1016/0550-3213(84)90374-2}{\emph{Nuclear
  Physics B} {\bfseries 247} (1984) 83}.

\bibitem{Resh-KZ}
N.~Reshetikhin, \emph{The knizhnik-zamolodchikov system as a deformation of the
  isomonodromy problem},
  \href{https://doi.org/10.1007/BF00420750}{\emph{Letters in Mathematical
  Physics} {\bfseries 26} (1992) 167}.

\bibitem{FMTV00}
G.~Felder, Y.~Markov, V.~Tarasov and A.~Varchenko, \emph{Differential equations
  compatible with kz equations},
  \href{https://doi.org/10.1023/A:1009862302234}{\emph{Mathematical Physics,
  Analysis and Geometry} {\bfseries 3} (2000) 139}.

\bibitem{Jimbo_2008}
M.~Jimbo, H.~Nagoya and J.~Sun, \emph{Remarks on the confluent kz equation for
  and quantum painlevé equations},
  \href{https://doi.org/10.1088/1751-8113/41/17/175205}{\emph{Journal of
  Physics A: Mathematical and Theoretical} {\bfseries 41} (2008) 175205}.

\bibitem{Nagoya_2010}
H.~Nagoya and J.~Sun, \emph{Confluent primary fields in the conformal field
  theory}, \href{https://doi.org/10.1088/1751-8113/43/46/465203}{\emph{Journal
  of Physics A: Mathematical and Theoretical} {\bfseries 43} (2010) 465203}.

\bibitem{Haghighat:2023vzu}
B.~Haghighat, Y.~Liu and N.~Reshetikhin, \emph{{Flat Connections from Irregular
  Conformal Blocks}},  \href{https://arxiv.org/abs/2311.07960}{{\ttfamily
  2311.07960}}.

\bibitem{Gu:2023plq}
X.~Gu, B.~Haghighat and K.~Loo, \emph{{Irregular Fibonacci Conformal Blocks}},
  \href{https://arxiv.org/abs/2311.13358}{{\ttfamily 2311.13358}}.

\bibitem{Gukov:2024yxa}
S.~Gukov, B.~Haghighat, Y.~Liu and N.~Reshetikhin, \emph{{Irregular KZ
  equations and Kac-Moody representations}},
  \href{https://arxiv.org/abs/2412.16929}{{\ttfamily 2412.16929}}.

\bibitem{Haghighat:2025qxu}
B.~Haghighat, \emph{{Liouville CFT, Matrix Models and constrained WZW}},
  \href{https://arxiv.org/abs/2508.06901}{{\ttfamily 2508.06901}}.

\bibitem{Gu:2026mlf}
X.~Gu, B.~Haghighat and P.~Putrov, \emph{{On the monodromy of KZ-equations with
  irregular singularities}},
  \href{https://arxiv.org/abs/2603.24098}{{\ttfamily 2603.24098}}.

\bibitem{Gaiotto:2013rk}
D.~Gaiotto and J.~Lamy-Poirier, \emph{{Irregular Singularities in the $H_3^+$
  WZW Model}},  \href{https://arxiv.org/abs/1301.5342}{{\ttfamily 1301.5342}}.

\bibitem{Ribault:2005wp}
S.~Ribault and J.~Teschner, \emph{{H+(3)-WZNW correlators from Liouville
  theory}}, \href{https://doi.org/10.1088/1126-6708/2005/06/014}{\emph{JHEP}
  {\bfseries 06} (2005) 014}
  [\href{https://arxiv.org/abs/hep-th/0502048}{{\ttfamily hep-th/0502048}}].

\bibitem{Gaiotto:2012sf}
D.~Gaiotto and J.~Teschner, \emph{{Irregular singularities in Liouville theory
  and Argyres-Douglas type gauge theories, I}},
  \href{https://doi.org/10.1007/JHEP12(2012)050}{\emph{JHEP} {\bfseries 12}
  (2012) 050} [\href{https://arxiv.org/abs/1203.1052}{{\ttfamily 1203.1052}}].

\bibitem{Kanno:2013vi}
H.~Kanno, K.~Maruyoshi, S.~Shiba and M.~Taki, \emph{{$W_3$ irregular states and
  isolated $N=2$ superconformal field theories}},
  \href{https://doi.org/10.1007/JHEP03(2013)147}{\emph{JHEP} {\bfseries 03}
  (2013) 147} [\href{https://arxiv.org/abs/1301.0721}{{\ttfamily 1301.0721}}].

\bibitem{Bonelli:2016qwg}
G.~Bonelli, O.~Lisovyy, K.~Maruyoshi, A.~Sciarappa and A.~Tanzini, \emph{{On
  Painlev\'e/gauge theory correspondence}},
  \href{https://doi.org/10.1007/s11005-017-0983-6}{\emph{Lett. Matth. Phys.}
  {\bfseries 107} (2017) 2359}
  [\href{https://arxiv.org/abs/1612.06235}{{\ttfamily 1612.06235}}].

\bibitem{Aganagic:2003qj}
M.~Aganagic, R.~Dijkgraaf, A.~Klemm, M.~Marino and C.~Vafa, \emph{{Topological
  strings and integrable hierarchies}},
  \href{https://doi.org/10.1007/s00220-005-1448-9}{\emph{Commun. Math. Phys.}
  {\bfseries 261} (2006) 451}
  [\href{https://arxiv.org/abs/hep-th/0312085}{{\ttfamily hep-th/0312085}}].

\bibitem{Drukker:2010jp}
N.~Drukker, D.~Gaiotto and J.~Gomis, \emph{{The Virtue of Defects in 4D Gauge
  Theories and 2D CFTs}},
  \href{https://doi.org/10.1007/JHEP06(2011)025}{\emph{JHEP} {\bfseries 06}
  (2011) 025} [\href{https://arxiv.org/abs/1003.1112}{{\ttfamily 1003.1112}}].

\bibitem{Frenkel:2015rda}
E.~Frenkel, S.~Gukov and J.~Teschner, \emph{{Surface Operators and Separation
  of Variables}}, \href{https://doi.org/10.1007/JHEP01(2016)179}{\emph{JHEP}
  {\bfseries 01} (2016) 179}
  [\href{https://arxiv.org/abs/1506.07508}{{\ttfamily 1506.07508}}].

\bibitem{Kostant1978}
B.~Kostant, \emph{On whittaker vectors and representation theory},
  \href{https://doi.org/10.1007/BF01410064}{\emph{Inventiones Mathematicae}
  {\bfseries 48} (1978) 287}.

\bibitem{Kostant1979}
B.~Kostant, \emph{Quantization and representation theory}, {\emph{Lecture Notes
  in Mathematics} {\bfseries 170} (1979) 287}.

\bibitem{AtiyahBott1983}
M.F.~Atiyah and R.~Bott, \emph{The yang--mills equations over riemann
  surfaces}, {\emph{Philosophical Transactions of the Royal Society of London
  Series A} {\bfseries 308} (1983) 523}.

\bibitem{GamayunIorgovLisovyy2013}
O.~Gamayun, N.~Iorgov and O.~Lisovyy, \emph{Conformal field theory of
  painlev\'e vi}, \href{https://doi.org/10.1007/JHEP10(2013)038}{\emph{Journal
  of High Energy Physics} {\bfseries 2013} (2013) }.

\bibitem{JimboMiwa1981_II}
M.~Jimbo and T.~Miwa, \emph{Monodromy preserving deformation of linear ordinary
  differential equations with rational coefficients. ii},
  \href{https://doi.org/10.1016/0167-2789(81)90021-X}{\emph{Physica D:
  Nonlinear Phenomena} {\bfseries 2} (1981) 407}.

\bibitem{vanSpaendonck:2022kit}
A.~van Spaendonck and M.~Vonk, \emph{{Painlev{\'e} I and exact WKB: Stokes
  phenomenon for two-parameter transseries}},
  \href{https://doi.org/10.1088/1751-8121/ac9e29}{\emph{J. Phys. A} {\bfseries
  55} (2022) 454003} [\href{https://arxiv.org/abs/2204.09062}{{\ttfamily
  2204.09062}}].

\end{thebibliography}\endgroup

\end{document}